\documentclass[lettersize,journal]{IEEEtran}

\IEEEoverridecommandlockouts

\usepackage{cite}
\usepackage{amsmath,amssymb,amsfonts}
\usepackage{amsthm}
\usepackage{bbm}
\usepackage{xcolor}
\usepackage{graphicx}
\usepackage{subcaption}
\usepackage{enumitem}
\usepackage{algorithm}
\usepackage{algorithmic}
\usepackage{tikz}
\usetikzlibrary{calc}

\newtheoremstyle{nonitalic}
{3pt}
{3pt}
{\normalfont}
{}
{\bfseries}
{.}
{.5em}
{}
\theoremstyle{nonitalic}

\newtheorem{theorem}{Theorem}
\newtheorem{lemma}{Lemma}
\newtheorem{remark}{Remark}

\newcommand\numberthis{\addtocounter{equation}{1}\tag{\theequation}}

\begin{document}

\title{\huge When Fresh Information Is Not Enough: Age-of-Actuation-Optimal Scheduling for\\ Energy-Harvesting Actuators}

\author{Ali Nikkhah, \IEEEmembership{Graduate Student Member, IEEE}, Anthony Ephremides, \IEEEmembership{Life Fellow, IEEE}, and Nikolaos Pappas, \IEEEmembership{Senior Member, IEEE}
\thanks{A. Nikkhah and N. Pappas are with the Department of Computer and Information Science, Linköping University, Linköping, Sweden, email: \{\texttt{ali.nikkhah, nikolaos.pappas\}@liu.se}. A. Ephremides is with the Electrical and Computer Engineering, University of Maryland, College Park, MD, USA, email: \texttt{etony@umd.edu}. \\ This work has been supported in part by the ELLIIT, and the European Union (6G-LEADER, 101192080). Part of this work has been submitted to the 2027 IEEE Information Theory Workshop (ITW)~\cite{nikkhah2027itw}.}
}

\maketitle
\begin{abstract}
We study transmission scheduling in an energy-harvesting actuation system where an actuation requires both a received update and available energy. While age of information (AoI) measures information freshness, age of actuation (AoA) measures the time elapsed since the latest completed actuation. The transmitter observes the receiver's \textcolor{black}{AoA} and battery state and decides whether to transmit a fresh update in each slot. We formulate average-AoA minimization as a Markov decision process \textcolor{black}{and establish} an AoA-threshold property for an optimal policy\textcolor{black}{. We also} derive the average AoA \textcolor{black}{in closed form} under the \textcolor{black}{AoI-optimal policy, which is always-transmit}. Numerical results show that AoA-aware scheduling can improve actuation freshness while reducing the number of transmissions, although information freshness may deteriorate. \textcolor{black}{The two policies coincide for some pairs of channel success and energy-arrival probabilities.} For a one-unit battery, the transmission rate \textcolor{black}{under the AoA-optimal policy} can vary nonmonotonically with the energy-arrival probability.
\end{abstract}

\section{Introduction} \label{Introduction}
\textcolor{black}{Timeliness is a key requirement in many communication systems, such as remote monitoring and networked control.
The most widely used metric for it is the Age of Information (AoI).
It was originally proposed in \cite{kaul2012real} and has since been analyzed in numerous communication systems \cite{yates2021age,pappas2023age}.
It is defined, at a destination, as the time elapsed since the generation of the freshest data available at that destination \cite{pappas2023age}, i.e., $I(t)=t-u(t)$, where $u(t)$ denotes the timestamp of the freshest received data, at or before time $t$. 
Unlike delay, it captures both how often updates are sent and how long they take to arrive.}

\textcolor{black}{However, AoI stops at reception.
It implicitly assumes that a received update is immediately useful.
In many systems, this is not the case.
The received data is used to trigger an action at the destination \cite{kountouris2021semantics}, and this action needs more than the data itself: a \textit{possibility}, an enabling condition to put data into use.
For example, an actuator powered by harvested energy can act only if its battery holds enough energy.
An update that arrives when the battery is empty resets the AoI, yet no action takes place.
The destination then holds fresh information, but its latest action is old.
Moreover, a policy that minimizes the AoI may spend transmissions on updates that never translate into actions.
Therefore, the freshness of the received data alone does not characterize the freshness of the actuations.}

Extensions of AoI have been proposed to capture semantic \textcolor{black}{attributes} and goal-oriented aspects of communication \textcolor{black}{\cite{uysal2022semantic,kountouris2021semantics}}, such as nonlinear aging costs \cite{kosta2020thecost}, version changes \cite{yates2021theage}, incorrect information \cite{maatouk2020theage}, and request- or query-based relevance \cite{yin2019onlythose,chiariotti2022query}.
\textcolor{black}{The closest to actuation is} Age of Loop (AoL), \textcolor{black}{which} extends AoI to closed-loop control by measuring the time since the generation of the state or command behind the latest applied action \cite{desantana2021age}.
\textcolor{black}{However, AoL assumes that actuation occurs periodically, with a fixed cycle.
It does not model actuation as a constrained possibility that depends on a random process, such as the availability of energy.}

\textcolor{black}{To capture this,} we proposed Age of Actuation (AoA) \cite{nikkhah2026age} and later Age of Actuated Information (AoAI) \cite{nikkhah2026semantics} to go one step beyond AoI and analyze the broader case in which \textit{data} requires a certain \textit{possibility} (e.g., physical or admission-related) at the destination to be effectively utilized, and to assess the freshness of this final outcome.
We refer to this concept as \textit{actuation}.
AoA is defined as the time elapsed since the most recent actuation at the destination\textcolor{black}{; its formal definition, together with that of the AoI, is given in Section~\ref{AoIandAoA}}.

\subsection{Related Work} \label{Related Work}
In a partial categorization, the analysis and optimization of AoI and, more generally, timeliness metrics can be divided into two main approaches: \textcolor{black}{randomized (state-independent) policies and dynamic (state-dependent) policies.}

\subsubsection{\textcolor{black}{Randomized Policies}} \label{state-independent randomized policies}
They typically enable tractable analytical characterization, often modeling the system as an uncontrolled Markov chain, and can be interpreted as open-loop control.
The decision policy is modeled as one or more random variables (RVs), usually Bernoulli RVs.
After characterizing the system dynamics, one typically derives a closed-form expression for the timeliness metric (e.g., average AoI) as a function of the corresponding probability parameters of these RVs.
In this approach, the policy does not depend on the system state; instead, at each time slot, decisions are made by sampling the RVs independently of the current system conditions \textcolor{black}{\cite{nikkhah2026age,nikkhah2026semantics,xiao2024infromation,zhao2025age,fountoulakis2022information}}. In our previous works, we utilized this approach \textcolor{black}{\cite{nikkhah2026age,nikkhah2026semantics}.}

\subsubsection{\textcolor{black}{Dynamic Policies}}
They incorporate feedback from the system state, with decisions made at each time slot based on the current state of the system, such as age, channel conditions, and resource availability, thereby corresponding to closed-loop control.
After decisions are made, the system transitions probabilistically to a new state according to its underlying dynamics.
This sequential decision-making and the stochastic nature of the system lead to the formulation of the problem as a Markov Decision Process (MDP), which provides a principled framework through the definition of states, actions, transition probabilities, and instant costs, typically correlating the cost with the age and minimizing the long-term cost which corresponds to the average AoI or related timeliness metrics.
\textcolor{black}{When the transition probabilities and instant costs are known, optimal policies} can be computed using techniques such as value iteration, relative value iteration, or policy iteration, which rely on recursively solving the Bellman optimality equations.
\textcolor{black}{This approach has been applied to AoI minimization in energy harvesting systems with multiple sources \cite{gindullina2021age}, wireless powered systems with joint sampling and updating \cite{abdelmagid2020aoi}, and wireless broadcast networks \cite{hsu2020scheduling}, and to version AoI in energy harvesting and caching systems \cite{salimnejad2025version,delfani2025version}.
AoI with an energy harvesting source has also been widely studied \cite{yates2015lazy,bacinoglu2015age,wu2018optimal,arafa2020age,krikidis2019average}.
Transmitting as soon as possible is not always age-optimal \cite{sun2016update}, and this also holds under energy harvesting \cite{yates2015lazy}.
The optimal policy has a threshold structure in the age, shown in discrete time \cite{bacinoglu2015age}, for a unit battery \cite{wu2018optimal}, and for a finite battery with thresholds that decrease with the stored energy \cite{arafa2020age}.
The average AoI of the greedy policy has also been derived in closed form \cite{krikidis2019average}.
In all these works, the optimized metric measures the freshness of the received data, not of its use.
Moreover, in \cite{yates2015lazy,bacinoglu2015age,wu2018optimal,arafa2020age,krikidis2019average}, energy is consumed by the transmission itself.
In our model, energy is instead needed at the destination for actuation.
As a result, the AoI-optimal policy is to always transmit, and the threshold structure appears in the AoA-optimal policy instead.}

State-dependent policies have been used to optimize the freshness of received data, but not the freshness of actuation.
On the other hand, the AoA has so far been studied only under state-independent randomized policies \cite{nikkhah2026age,nikkhah2026semantics}.
Thus, it is not known what an AoA-optimal state-dependent policy looks like, and whether and when it differs from the AoI-optimal policy.
This paper addresses this gap.

\subsection{Contribution}
In this work, we consider a system model that involves state-dependent decision-making to transmit data. We analyze the average AoA and the average AoI under both AoA-Optimal (AoA-Opt.) policy and AoI-Optimal (AoI-Opt.) policy. We thoroughly compare the outcomes and illustrate the trade-offs in timeliness and cost from the perspective of greedy and non-greedy regions, and show how the AoA captures the freshness of completed actuations and can therefore lead to a different optimal transmission policy from the AoI. The main contributions of the paper are as follows:
\begin{itemize}
    \item We characterize the AoA as an average-cost-per-stage problem and derive the optimal average AoA and the corresponding optimal policies using Relative Value Iteration (RVI).

    \item We investigate structural properties of the MDP and prove \textcolor{black}{an AoA-threshold structure for the AoA-Opt.\ policy}, while such a structure for the battery state is observed only numerically.

    \item We characterize the AoA for the AoI-Opt.\ \textcolor{black}{(always-transmit)} policy using a Markov chain formulation, derive its steady-state behavior, and obtain a closed-form expression for the average AoA under \textcolor{black}{this} policy. \textcolor{black}{We also establish the ordering among the four values of the average AoA and the average AoI under the two policies.}

    \item We present and analyze two sets of operating regions, obtained by comparing the AoA-optimal and AoI-optimal policies, along with the resulting optimal and suboptimal values of the average AoA and the average AoI.

    \item We \textcolor{black}{present counterintuitive} system behaviors that emerge at very low battery capacity, namely a battery capacity of one\textcolor{black}{, such as a transmission rate that varies nonmonotonically with the energy-arrival probability.}
\end{itemize}

\begin{figure}[t]
\centering
\includegraphics[width=\linewidth]{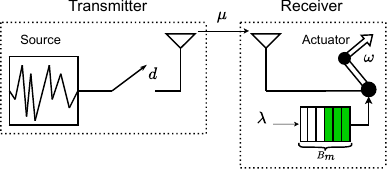}
\caption{The considered system model, where successful data reception and energy availability jointly trigger actuation.}
\label{fig:System Model}
\end{figure}

\section{System Model} \label{sec:System Model}
\textcolor{black}{In this section, we first describe the system dynamics, i.e., the interaction of data transmission, channel success, energy arrivals, the battery, and the actuator over time, and then define the considered timeliness metrics, the AoI and the AoA.}
\textcolor{black}{\subsection{System Dynamics} \label{System Dynamics}}
Fig. \ref{fig:System Model} illustrates the considered system model.
Time is divided into discrete slots of equal length.
At the transmitter side, there is a time-varying information source containing the data required for actuation.
\textcolor{black}{At the beginning of} each time slot, the transmitter decides whether or not to transmit the data.
\textcolor{black}{This decision is based on the state of the system at the end of the previous time slot.}
\textcolor{black}{The transmission of data} entails sampling the source, generating a packet of information, and transmitting it.
\textcolor{black}{This} can also be regarded as issuing a command to actuate.
We denote the transmission action by $d(k) \in D=\{0,1\}$, where $d(k)=0$ denotes no transmission and $d(k)=1$ denotes transmission, in time slot $k$.
Given that transmission occurs, the channel successfully delivers the packet to the receiver with probability (w.p.) $\mu$ and fails to do so w.p. $\bar{\mu}=1-\mu$.
We express its related random variable (RV) by $M(k)\textcolor{black}{\in}\{0,1\}$ where $M(k)=0$ denotes the event of failure of the channel to deliver the transmitted data packet and $M(k)=1$ denotes the event of success in delivering the data packet, in time slot $k$.
At the receiver's end, there is an actuator that actuates based on the data it receives from the source.
This actuation requires the presence of both \textit{data} and \textit{energy}.
The energy is provided from a battery that stores discrete energy units, with a maximum capacity of $B_m$.
Each energy unit is sufficient to power one actuation.
An energy unit can be utilized within the same time slot in which it is harvested.
The battery harvests one energy unit w.p. $\lambda$, in each time slot.
\textcolor{black}{If the battery is full, a newly harvested unit is discarded unless an actuation in the same slot frees one unit of capacity.}
We define the RV representing energy arrival (harvesting) as $\Lambda(k) \in \{0,1\}$, where $\Lambda(k)=0$ denotes no energy arrival and $\Lambda(k)=1$ denotes the arrival of one energy unit in time slot $k$.
We denote the battery state by $B(k) \in \{0,1,2,\dots,B_{m}\}$, which represents the number of energy units stored in the battery at the end of time slot $k$.

The following relation describes the evolution of the battery state at the receiver side based on the dynamics discussed so far\footnote{We refrain from writing the first case as $\max\{B(k-1)-1,0\}$, because when $\omega(k)=1$ and $\Lambda(k)=0$, the battery could not have been empty (i.e., $B(k-1)\neq 0$) and must have had at least one energy unit.}
\begin{equation} \label{Battery_Evolution}
B(k)=
\begin{cases}
B(k-1)-1  & \omega(k)=1 \land \Lambda(k)=0,\\
B(k-1) & (\omega(k)=0 \land \Lambda(k)=0) \lor \\
&(\omega(k)=1 \land \Lambda(k)=1),\\
\min\{B(k-1)+1,\\
\quad \ \ \ B_m\} & \omega(k)=0 \land \Lambda(k)=1,
\end{cases}
\end{equation}
where $\omega(k) \in \Omega=\{0,1\}$ and $\omega(k)=0$ denotes the absence of an actuation and $\omega\textcolor{black}{(k)}=1$ denotes the occurrence of an actuation in \textcolor{black}{time slot} $k$.

The actuator performs an actuation in a given time slot if the transmitter transmits a data packet, the channel successfully delivers it, and either the battery is not empty or, if the battery is empty, it harvests one energy unit within the same time slot. It can be expressed as
\begin{equation}  \label{Actuation_Condidtion_S12}
\omega(k)=
\begin{cases}
1 &   (d(k)=1) \land (M(k)=1) \land \\ 
&( (B(k-1) \neq 0) \lor (\Lambda(k)=1) ),\\
0 & \text{otherwise}.
\end{cases}
\end{equation}

\textcolor{black}{
\begin{remark}
The precise chronology of events and measurements within a time slot is as follows:
At the beginning of each time slot, the action $d$, whether to transmit data, is taken, based on the system state (age and battery level) from the end of the previous time slot. Then the channel succeeds or fails to deliver the data packet, $M$; an energy arrival occurs or not, $\Lambda$; and, if possible, an actuation is accomplished, $\omega$. All the aforementioned events occur in rapid succession, with a negligible time difference relative to the slot length, while preserving their causal order. At the end of the time slot, the battery level $B$ and the age $A$ are measured, and together they constitute the state for this time slot.
\end{remark}
}

\textcolor{black}{
\begin{remark}
The state-dependent policy requires knowledge of the current state when taking any action. In our system model, at the beginning of each time slot, the transmitter's action \textit{depends} on the AoA and the battery level at the receiver from the end of the previous time slot. Therefore, we assume there is a separate noiseless \textit{feedback} link from the receiver to the transmitter.
\end{remark}
 }

\subsection{Considered Timeliness Metrics: AoI and AoA} \label{AoIandAoA}
We defined AoI and AoA, in general, in Section \ref{Introduction}. Here we define them in more detail corresponding to this paper.

For a discrete-time and no-queuing-before-reception model, as considered in this paper, where packets experience no delay, and each received update is the freshest possible at the time of reception and a potential actuation, the evolution of the AoI and AoA are given by
\begin{equation} \label{AoI_Evolution}
I(k)=
\begin{cases}
1 & v_{d}(k)=1,\\
I(k-1)+1 & v_{d}(k)=0,
\end{cases}
\end{equation}
and
\begin{equation}
A(k)=
\begin{cases}
1 & \omega(k)=1,\\
A(k-1)+1 & \omega(k)=0,
\end{cases}
\end{equation}
respectively, where $v_{d}(k) \in \mathcal{V}_{d} = \{0,1\}$, with $v_{d}(k)=1$ indicating successful reception of data and $v_{d}(k)=0$ indicating no data reception, at time slot $k$, defined as
\begin{equation} \label{Data_Reception}
v_{d}(k)=
\begin{cases}
1 & d(k)=1 \land M(k)=1,\\
0 & \text{otherwise}.
\end{cases}
\end{equation}
An actuation requires both \textit{data} and the \textit{possibility}; in this work, the possibility corresponds to the availability of energy. Thus, for the actuation we have
\begin{equation} \label{Actuation_Considered_Metrics}
\omega(k)=
\begin{cases}
1 & v_{d}(k)=1 \land v_{e}(k)=1,\\
0 & \text{otherwise},
\end{cases}
\end{equation}
where $v_{e}(k) \in \mathcal{V}_{e} = \{0,1\}$, with $v_{e}(k)=0$ denoting the absence of energy, and $v_{e}(k)=1$ denoting the availability of energy in time slot $k$, defined as
\begin{equation} \label{Energy_Availability}
v_{e}(k)=
\begin{cases}
1 & B(k-1) \textcolor{black}{\geq 1} \lor \Lambda(k)=1,\\
0 & B(k-1)=0 \land \Lambda(k)=0.
\end{cases}
\end{equation}
Note that placing (\ref{Data_Reception}) and (\ref{Energy_Availability}) into (\ref{Actuation_Considered_Metrics}) results in (\ref{Actuation_Condidtion_S12}).
 
\section{Age of Actuation (AoA) Optimization} \label{Age of Actuation (AoA) Optimization}
In this section, we characterize a standard Markov decision process (MDP) problem to optimize the average AoA, defined in Section \ref{AoIandAoA}, and obtain its optimal values and related optimal policies via the relative value iteration (RVI) algorithm. At the end of this section, we discuss the average AoI and two other averages under the AoA-Opt. policy in a remark.
\textcolor{black}{By ${v}^{*,A}$ (or simply ${v}^{*}$)\footnote{We reserve the unspecified superscript for the AoA-optimal policy, since the primary focus of this paper is on optimizing the average AoA.}, we denote the value of $v$ under the AoA-Opt. policy.}

\subsection{MDP Characterization} \label{MDP Characterization}
The MDP to characterize the dynamics of the problem has the following components:

\noindent
\textbf{States}: The state space of the problem is $S$, whose elements are $s_{(A,B)}(k)=(A(k),B(k))$, where \textcolor{black}{$A(k) \in \{1,2,\dots,A_m\}$} denotes the AoA\footnote{\textcolor{black}{Originally, we have $A(k) \in \mathbb{N}$; however, for practically solving the Bellman equation of the MDP problem via the relative value iteration (RVI) algorithm, we must truncate it to a finite range. A sufficiently large value of $A_m$ is chosen so that it effectively represents infinity.}}, and $B(k) \in \{0,1,2,\dots,B_m\}$ denotes the number of energy units available in the battery, both measured at the end of time slot $k$. \textcolor{black}{Thus, the total number of states is $|S| = A_{m} \times (B_{m}+1)$.} We use $A_{s}$ and $B_{s}$ to denote the AoA and battery level associated with state $s$, respectively.

\noindent
\textbf{Actions}: The action space is $d_{s}\in D=\{0,1\}, \forall s \in S$, where $d=0$ implies no transmission and $d=1$ implies transmission.

\noindent
\textbf{Instant Cost}:
We define $c_{s,d}=c_{s_{(A,B)},d}$ as the instant cost incurred when taking action $d$ in state $s$ and define it as
\begin{equation} \label{cost_for_notrans_and_trans_red}
c_{s_{(A,B)},d}=c_{s_{(A,B)}}=A.
\end{equation}
That is, the instant cost of taking any action $d$ in state $s_{(A,B)}$ equals the AoA of that state, $A_s$\textcolor{black}{, and is therefore independent of the action}.

\noindent
\textbf{System Knowledge (dynamics)}: The joint system evolution is described for three battery levels. If $B(k)=0$, the corresponding relations are
\begin{equation} \label{equ:Sys_dynamics_B=B_0}
\begin{aligned}
s_{A,B}(k+1)=
\begin{cases}
(1,0) &   \ d(k+1)=1 , \text{w.p.} \ w,\\
(\textcolor{black}{A^{\text{tr}}(k)},0) &  \  d(k+1)=1 , \text{w.p.} \ \bar{\lambda},\\
(\textcolor{black}{A^{\text{tr}}(k)},1) & \ d(k+1)=1 , \text{w.p.} \ y,\\
(\textcolor{black}{A^{\text{tr}}(k)},1) &  \ d(k+1)=0 , \text{w.p.} \ \lambda,\\
(\textcolor{black}{A^{\text{tr}}(k)},0) &  \ d(k+1)=0 , \text{w.p.} \ \bar{\lambda},\\
\end{cases}
\end{aligned}
\end{equation}
if $0<B(k)<B_m$, they are
\begin{equation} \label{equ:Sys_dynamics_0<B<B_m}
\begin{array}{l}
s_{A,B}(k+1)= \\[2pt]
\begin{cases}
(1,B(k)) &   \ d(k+1)=1 , \text{w.p.} \ w,\\
(1,B(k)-1) &\  d(k+1)=1 , \text{w.p.} \ x,\\
(\textcolor{black}{A^{\text{tr}}(k)},B(k)) &   \ d(k+1)=1 , \text{w.p.} \ z,\\
(\textcolor{black}{A^{\text{tr}}(k)},B(k)+1) &  \ d(k+1)=1 , \text{w.p.} \ y,\\
(\textcolor{black}{A^{\text{tr}}(k)},B(k)+1) &  \ d(k+1)=0 , \text{w.p.} \ \lambda,\\
(\textcolor{black}{A^{\text{tr}}(k)},B(k)) &  \ d(k+1)=0 , \text{w.p.} \ \bar{\lambda},\\
\end{cases}
\end{array}
\end{equation}
and if $B(k)=B_m$, they are
\begin{equation} \label{equ:Sys_dynamics_B=B_m}
\begin{aligned}
s_{A,B}(k+1)=
\begin{cases}
(1,B_{m}) &    \ d(k+1)=1 , \text{w.p.} \ w,\\
(1,B_{m}-1) &  \  d(k+1)=1 , \text{w.p.} \ x,\\
(\textcolor{black}{A^{\text{tr}}(k)},B_{m}) &  \ d(k+1)=1 , \text{w.p.} \ \bar{\mu},\\
(\textcolor{black}{A^{\text{tr}}(k)},B_{m}) &  \ d(k+1)=0 , \text{w.p.} \ 1,\\
\end{cases}
\end{aligned}
\end{equation}
\textcolor{black}{where $A^{\text{tr}}(k)=\min\{A(k)+1,A_{m}\}$, reflecting the age truncation used to keep the MDP problem finite-sized}\textcolor{black}{, and $w=\mu\lambda$, $x=\mu\bar{\lambda}$, $y=\bar{\mu}\lambda$, and $z=\bar{\mu}\bar{\lambda}$}.

\subsection{Problem Formulation}
\noindent
The average AoA is
\begin{equation} \label{equ:Average_AoA_MDP}
\bar{A}^{\pi}=\limsup_{N \rightarrow \infty} \frac{1}{N+1} \sum_{k=0}^{N} \mathbb{E}_{P}^{\pi}[A(k)|s_{(A(0),B(0))}],
\end{equation}
for the initial state\footnote{In Theorem~\ref{theo:WAUniErg}, we prove that under the AoA-Opt. policy, the optimal solution is independent of and achieved from any initial state.} $s_{(A(0),B(0))}$.
The expectation is taken over the system's inherent stochasticity, $P$ (i.e., system knowledge), under the policy $\pi$ (i.e., action taken by the transmitter).
A policy $\pi$ is defined as a mapping from the state space to the action space.
The policy specifies the action to be taken at each state
\begin{equation}
\pi = \{d_s\}, \forall s,
\end{equation}
and $\Pi$ is the set of all possible policies.
Moreover, $P$ denotes the set of all probabilities in the system
\begin{equation}
P=\{p_{s,d_s,s'}\}, \ \forall s,s' \in S, \ d_s \in D,
\end{equation}
which is referred to as system knowledge, where $p_{s,d,s'}$ denotes the probability of transitioning to state $s'$ given that the system is in state $s$ and decision $d$ is taken. We define
\begin{equation}
\pi^{*}=\arg \min_{\pi \in \Pi} \bar{A}^{\pi}, \ \ \bar{A}^{*}=\min_{\pi \in \Pi} \bar{A}^{\pi}
\end{equation}
as the AoA-Opt. policy, $\pi^{*}$, minimizing the average AoA and the corresponding achieved optimal average AoA, respectively.

The combination of stochasticity and decision-making leads to an MDP.
We also define the transmission rate under the AoA-Opt. policy by
\begin{align*} 
\bar{D}^{*,A}&=\limsup_{N \rightarrow \infty} \frac{1}{N+1} \sum_{k=0}^{N} \mathbb{E}_{P}^{\pi^*}[\mathbbm{1}_{\{d_{s(k)}^{*}=1\}}|s_{(A(0),B(0))}]\\
&=\sum_{s \in S}\text{Pr}\{s_{d_{s}^{*}=1}\}, \numberthis \label{equ:active_fraction}
\end{align*}
which represents the long-run fraction of time during which the transmitter is scheduled to transmit, equivalently the probability of transmission, under the AoA-Opt. policy. This quantity also captures the cost of transmission associated with the AoA-Opt. policy. Finally, we define the probability of an empty battery under the AoA-Opt. policy as
\begin{align*} 
\bar{B}_{0}^{*,A}&=\limsup_{N \rightarrow \infty} \frac{1}{N+1} \sum_{k=0}^{N} \mathbb{E}_{P}^{\pi^*}[\mathbbm{1}_{\{B_{s(k)}=0\}}|s_{(A(0),B(0))}]\\
&=\sum_{s \in S}\text{Pr}\{s_{(A,0)}\}, \numberthis \label{equ:empty_battery}
\end{align*}
The existence of the probabilities (\ref{equ:active_fraction}) and (\ref{equ:empty_battery}) is discussed in Theorem~\ref{theo:WAUniErg}.


\subsection{Optimality Analysis}
Our problem falls under the infinite-horizon\footnote{The infinite-horizon setting implies that there is no terminal state.}, finite-state, finite-action setting. The corresponding Bellman equation is given by \cite[Proposition 5.2.1]{bertsekas2012dynamic}
\begin{equation} \label{Bellman_general}
C^{*} + h^{*}_{s} = \min_{d \in D}[c_{s,d}+\sum_{s^{'} \in S}p_{s,d,s^{'}}h_{s^{'}}^{*}], \ \forall s \in S,
\end{equation}
in which $h_{s}^{*}$ is defined as the optimal differential cost function\footnote{In the remainder of the paper, to avoid ambiguity, we drop the term “differential” and simply refer to it as the cost function.} associated with state $s$, for all $s \in S$. The quantity $C^*$ denotes the optimal average cost per stage of the MDP.
\textcolor{black}{Note} that
\begin{align*}
C^{*}&= \lim_{N \rightarrow \infty} \frac{1}{N+1}\sum_{k=0}^{N} c_{s_{(A,B)},d_{s}^{*}}(k), \\
&= \lim_{N \rightarrow \infty} \frac{1}{N+1}\sum_{k=0}^{N} c_{s_{(A,B)}}(k), \\
&= \lim_{N \rightarrow \infty} \frac{1}{N+1}\sum_{k=0}^{N} A(k),\\
&= \bar{A}^{*}, \numberthis \label{proof_opt_cost_is_ave_age_red}
\end{align*}
where $c_{s_{(A,B)},d_{s}^{*}}$ is the instant cost incurred when action $d$ of the optimal policy is taken in state $s$.
As discussed in Section \ref{MDP Characterization}, this cost is independent of the action and thus of the policy, and it equals the age of the current state.

Using (\ref{cost_for_notrans_and_trans_red}) and (\ref{proof_opt_cost_is_ave_age_red}), we can rewrite the Bellman equation (\ref{Bellman_general}) as
\begin{equation}  \label{Bellman_our_case_red}
\bar{A}^{*} + h^{*}(s) = A_{s} + \min_{d \in D}[\sum_{s^{'} \in S}p_{s,d,s^{'}}h^*(s^{'})],\ \forall s \in S.
\end{equation}
\textcolor{black}{Based on equations (\ref{equ:Sys_dynamics_B=B_0}), (\ref{equ:Sys_dynamics_0<B<B_m}), and (\ref{equ:Sys_dynamics_B=B_m}),} we can specify (\ref{Bellman_our_case_red}) for states with different battery levels. For states with $B=0$, it is

\begin{align*}
\bar{A}^{*} + h^{*}_{(A,0)} = A+ \min[ & w h^{*}_{(1,0)} +\bar{\lambda} h^{*}_{(\textcolor{black}{A^{\text{tr}}},0)}+ y h^{*}_{(\textcolor{black}{A^{\text{tr}}},1)},\\ & \bar{\lambda} h^{*}_{(\textcolor{black}{A^{\text{tr}}},0)} + \lambda h^{*}_{(\textcolor{black}{A^{\text{tr}}},1)}], \numberthis  \label{equ:BellmanEqu_B_0}
\end{align*}
for states with $0<B< B_{m}$, it is
\begin{align*}
\bar{A}^{*} + &h^{*}_{(A,B)} = A+ \min[ w h^{*}_{(1,B)} + x h^{*}_{(1,B-1)} + \\ & y h^{*}_{(\textcolor{black}{A^{\text{tr}}},B+1)} +z h^{*}_{(\textcolor{black}{A^{\text{tr}}},B)},  \bar{\lambda} h^{*}_{(\textcolor{black}{A^{\text{tr}}},B)} + \lambda h^{*}_{(\textcolor{black}{A^{\text{tr}}},B+1)} ], \numberthis  \label{equ:BellmanEqu_B_0_B_B_m}
\end{align*}
and for states with $B=B_{m}$, it is
\begin{align*}
\bar{A}^{*} + h^{*}_{(A,B_{m})} = A+ \min[& w  h^{*}_{(1,B_{m})} +x h^{*}_{(1,B_{m}-1)}+ \\
& \bar{\mu} h^{*}_{(\textcolor{black}{A^{\text{tr}}},B_{m})},   h^{*}_{(\textcolor{black}{A^{\text{tr}}},B_{m})}  ],  \numberthis  \label{equ:BellmanEqu_B_m}
\end{align*}

where $h^{*}_{(A,B)}\equiv h^{*}(s_{(A,B)})$ is the optimal cost function associated with state $s_{(A,B)}$ \textcolor{black}{and $A^{\text{tr}}=\min\{A+1,A_{m}\}$ due to age truncation.} The problem is then solved using the RVI algorithm, as described in Algorithm \ref{alg:RVI}. As established in Theorem~\ref{theo:WAUniErg}, the RVI algorithm \textcolor{black}{converges and yields the optimal average cost} for (\ref{equ:Average_AoA_MDP}) that is independent of the initial state.
\begin{algorithm}
\caption{Relative Value Iteration (RVI) for Average Cost MDP}
\label{alg:RVI}
\begin{algorithmic}[1]
\STATE \textbf{Input:} State space $S$, action space $D$, \textcolor{black}{instant costs} $c_{s}$, transition probabilities $p_{s,d,s^{'}}$, tolerance $\epsilon > 0$
\STATE \textbf{Output:} Optimal relative values $h^{*}(s)$, average cost $\bar{A}^*$, optimal policy $\pi^*$
\STATE Initialize $h(s), h_{\text{new}}(s) \leftarrow 0, \ \forall s \in S$
\STATE Choose $s_{(1,0)}$ as the reference state
\REPEAT
    \STATE $h(s) \leftarrow h_{\text{new}}(s), \ \forall s$
    \FOR{each state $s \in S$}
        \STATE $h_{\text{new}}(s) \leftarrow A_{s} + \min_{d \in D}\left[\sum_{s^{'} \in S} p_{s,d,s^{'}}  h(s^{'})\right]$
    \ENDFOR
    \STATE Normalize: $h_{\text{ref}} \leftarrow h_{\text{new}}(s_{(1,0)})$  
    \STATE \hphantom{Normalize:} $h_{\text{new}}(s) \leftarrow h_{\text{new}}(s) - h_{\text{ref}}, \ \forall s$
\UNTIL{$\max_{s \in S} |h(s) - h_{\text{new}}(s)| < \epsilon$}
\STATE $h^{*}(s) \leftarrow h_{\text{new}}(s), \ \forall s$
\STATE $\bar{A}_{s}^{*} = A_{s}+ \min_{d \in D} \left[\sum_{s^{'} \in S} p_{s,d,s^{'}}  h^*(s^{'}) \right] - h^{*}(s), \ \forall s $
\STATE $\bar{A}^{*}=\frac{1}{|S|} \sum_{s \in S} \bar{A}_{s}^{*}$
\STATE $d_{s}^* = \arg\min_{d \in D}\left[\sum_{s^{'} \in S} p_{s,d,s^{'}}  h^*(s^{'})\right], \ \forall s$
\end{algorithmic}
\end{algorithm}

\begin{remark}
The average AoI, the transmission rate, and the probability of an empty battery under the AoA-Opt. policy are calculated by simulating the system under the obtained optimal policy for each pair of $\mu$ and $\lambda$. They are denoted as $\bar{I}^{*,A}$, $\bar{D}^{*,A}$, and $\bar{B}_{0}^{*,A}$, respectively.
\end{remark}

\begin{theorem} \label{theo:WAUniErg}
For the MDP problem presented in Section \ref{MDP Characterization}, \textcolor{black}{with $0<\mu, \lambda<1$,} the following statements are true:

\begin{enumerate}[label=(\roman*)]

\item \label{theo:WAUniErg_unichain} All the policies, including the optimal ones, of the MDP are unichain.

\item \label{theo:WAUniErg_WA} The MDP problem presented is weakly accessible.

\item \label{theo:WAUniErg_Erg} The induced Markov chains from all policies, including the optimal ones, have exactly one \textit{ergodic} recurrent class.

\item \label{theo:WAUniErg_Indep} \textcolor{black}{Algorithm~\ref{alg:RVI} (RVI) converges and yields the optimal average cost (equivalent to the average AoA), which is independent of the starting state.}

\item \label{theo:WAUniErg_Irre} \textcolor{black}{If $A_{m}$ is large enough that the optimal policy has $d=1$ in at least one state at every battery level, then the induced Markov chains from the optimal policies have no transient states and are fully irreducible.}

\item \label{theo:WAUniErg_Exist} The probabilities (\ref{equ:active_fraction}) and (\ref{equ:empty_battery}) exist and are independent of the starting state.
\end{enumerate}
\end{theorem}
\begin{proof}
See Appendix \textcolor{black}{\ref{App:WAUniErg}}.
\end{proof}

\subsection{Structural Property of the Result\textcolor{black}{s}}
In this section, we discuss the threshold structure of the optimal AoA-Opt.\ policy.

\begin{lemma}\label{lemma:increasing}
The cost function of $h_{(A,B)}$ at each iteration of the RVI is
\begin{enumerate}[label=(\roman*)]
\item \label{lemma:increasing_A} increasing with $A$
\begin{equation} \label{equ:increasing_A}
A_{1} \leq A_{2} \rightarrow h_{(A_{1},B)}^{k} \leq h_{(A_{2},B)}^{k}, \ \forall k \geq 0,
\end{equation}
\item \label{lemma:increasing_B} decreasing with $B$
\begin{equation}\label{equ:increasing_B}
B_{1} \leq B_{2} \rightarrow h_{(A,B_{1})}^{k} \geq h_{(A,B_{2})}^{k}, \ \forall k \geq 0,
\end{equation}
\end{enumerate}
including for the converged optimal cost functions, $h^{*}$.
\end{lemma}
\begin{proof}
Refer to Appendix \ref{App:lemma:increasing}.
\end{proof}

\begin{lemma} \label{lemma:chain_inequality}
Consider Lemma \ref{lemma:increasing}\ref{lemma:increasing_A} and \ref{lemma:increasing}\ref{lemma:increasing_B}. If $s_{(A_{1},B_{1})}$ and $s_{(A_{2},B_{2})}$ are two states such that $A_{1} \leq A_{2}$ and $B_{1} \geq B_{2}$, then
\begin{equation} \label{equ:lemma_chain_main}
h_{(A_{1},B_{1})}^{k} \leq  h_{(A_{2},B_{2})}^{k}.
\end{equation}
\end{lemma}
\begin{proof}
\begin{equation} \label{equ:lemma_chain_proof}
h_{(A_{1},B_{1})}^{k} \overset{(a)}{\leq} h_{(A_{2},B_{1})}^{k} \overset{(b)}{\leq} h_{(A_{2},B_{2})}^{k},
\end{equation}
where (a) and (b) follow from Lemmas \ref{lemma:increasing}\ref{lemma:increasing_A} and \ref{lemma:increasing}\ref{lemma:increasing_B}, respectively.
\end{proof}

\begin{theorem} \label{theo:threshold_A}
The optimal policy is threshold-based w.r.t. the AoA ($A$). That is, for two states $s_{1}\equiv s_{(A_{1},B)}$ and $s_{2}\equiv s_{(A_{2},B)}$, then
\begin{equation}  \label{structure_A}
\forall B: \
d_{2}^{*}=
\begin{cases}
1 &  (d_{1}^{*}=1) \land (A_{1} \leq A_{2}),\\
0 &  (d_{1}^{*}=0) \land (A_{1} \geq A_{2}).
\end{cases}
\end{equation}
\end{theorem}
\begin{proof}
The proof can be found in the Appendix \ref{App:threshold_A}.
\end{proof}
\textcolor{black}{
\begin{remark} \label{theo:threshold_B}
Similarly, we investigate a threshold-based structure of the optimal policy w.r.t. the battery level ($B$), in Appendix \ref{App:threshold_B}. But the monotonicity results do not establish a threshold-based structure along the battery level dimension. That is, given two states $s_{1}\equiv s_{(A,B_{1})}$ and $s_{2}\equiv s_{(A,B_{2})}$, the threshold structure
\begin{equation} \label{structure_B}
\forall A: \
d_{2}^{*}=
\begin{cases}
1 &  (d_{1}^{*}=1) \land (B_{1} \leq B_{2}),\\
0 &  (d_{1}^{*}=0) \land (B_{1} \geq B_{2}),
\end{cases}
\end{equation}
cannot be established.
However, such a structure is consistently observed in numerical results. Thus, it remains an analytical question for the present model.
\end{remark}}



\section{Age of Information (AoI) Optimization} \label{Age of Information (AoI) Optimization}
In this section, we analyze the system outcomes when the goal is to optimize the average AoI.
\textcolor{black}{By ${v}^{*,I}$, we denote the value of $v$ under the AoI-Opt. policy.}

Intuitively, in the considered system model in Fig. \ref{fig:System Model}, the optimal action to minimize the average AoI is to transmit at all time slots: $d_{s}^{*,I}=1, \ \forall s \in S$. Consequently, the transmission rate is $\bar{D}^{*,I}=1$, the average AoI is $\bar{I}^{*,I}=1/\mu$ \cite{pappas2023age}, and the probability of an empty battery is \textcolor{black}{$\bar{B}_{0}^{*,I}=(1-r)/(1-r^{B_{m}+1})$ for $\mu \neq \lambda$, and $\bar{B}_{0}^{*,I}=1/(B_{m}+1)$ for $\mu=\lambda$, where $r=(\bar{\mu}\lambda)/(\mu\bar{\lambda})$ \cite{srikant2013communication}.}
\textcolor{black}{\begin{remark}
The always-transmit policy is optimal for average AoI because each transmission carries the freshest possible sample, transmission attempts do not create a queue or interfere with future opportunities, and transmissions have no cost or constraint in the AoI optimization. This can be verified by solving a separate, trivial MDP where the state is only the AoI value at each time slot; due to its size and triviality, we have opted not to include it in the paper. If any of these assumptions change, however, the always-transmit policy need not remain optimal.
\end{remark}}
However, the average AoA under this policy is part of the contributions of this work and is stated in the following theorem:
\begin{theorem}
The average AoA for the AoI-Opt. policy with $\mu$ and $\lambda$ (as service probability and arrival probability, respectively, regarding the battery queue) and the battery size of $B_{m}$, \textcolor{black}{in general,} is
\begin{equation} \label{equ:Average_AoA_MC_body_MOST_GENERAL}
\bar{A}^{*,I}=\frac{1}{\mu} + \frac{\bar{\lambda}}{\lambda} \textcolor{black}{\bar{B}_{0}^{*,I}},   
\end{equation}
\textcolor{black}{which, for $\mu \neq \lambda$, becomes}
\begin{equation} \label{equ:Average_AoA_MC_body_general}
\bar{A}^{*,I}=\frac{\mu^{B_{m}+2}\left( 1- \lambda \right)^{B_{m}+1}-\left( 1- \mu \right)^{B_{m}+1} \lambda^{B_{m}+2}}{\mu^{B_{m}+2}\lambda\left( 1- \lambda \right)^{B_{m}+1}-\mu\left(1-\mu\right)^{B_{m}+1} \lambda^{B_{m}+2}},
\end{equation}
\textcolor{black}{and, for $\mu=\lambda$, becomes}
\begin{equation} \label{equ:Average_AoA_MC_body_equal}
\bar{A}^{*,I}=\frac{2-\mu+B_{m}}{\mu+\mu B_{m}}.
\end{equation}
\end{theorem}
\begin{proof}
The proof is available in the Appendix \ref{App:Randomized_Policy_Proof}.
\end{proof}

\begin{remark} \label{remark:AoA_AoI_Opt_interpretation}
\textcolor{black}{Note that in (\ref{equ:Average_AoA_MC_body_MOST_GENERAL}), the first term is the average AoI under the AoI-Opt.\ policy, $\bar{I}^{*,I}$, and the second term is a penalty proportional to the probability of an empty battery. Hence, under this policy, the gap between the average AoA and the average AoI is caused only by the empty battery.
Moreover, (\ref{equ:Average_AoA_MC_body_MOST_GENERAL}) gives $1/\lambda$ for $\mu=1$, since then $\bar{B}_{0}^{*,I}=1$, and $1/\mu$ for $\lambda=1$, since then the penalty vanishes. For $\mu=\lambda=1$, it gives $1$. These confirm the expected and known results.}
\end{remark}

\begin{remark} \label{remark:AoI_Opt_Dynamics}
The expressions characterizing the dynamics of the system under the AoI-Opt. policy, \textcolor{black}{which are utilized in Appendix \ref{App:Randomized_Policy_Proof} for obtaining $\bar{A}^{*,I}$,} are direct reductions of expressions (\ref{equ:Sys_dynamics_B=B_0}), (\ref{equ:Sys_dynamics_0<B<B_m}), and (\ref{equ:Sys_dynamics_B=B_m}), where only the cases with $d(k)=1$ are relevant, \textcolor{black}{with omission of the terms $d(k)=1$ themselves due to redundancy.} A Markov chain can describe this dynamic.
\end{remark}

\begin{theorem} \label{theo:relation_4_values}
\textcolor{black}{The four values of average AoA and average AoI under the AoA-Opt.\ and AoI-Opt.\ policies satisfy
\begin{equation} \label{equ:relation_4_values}
\bar{I}^{*,I} \leq \bar{I}^{*,A} \leq \bar{A}^{*,A} \leq \bar{A}^{*,I}.
\end{equation}}
\end{theorem}
\begin{proof}
\textcolor{black}{We have i) $\bar{I}^{*,I} \leq \bar{I}^{*,A}$ and ii) $\bar{A}^{*,A} \leq \bar{A}^{*,I}$, since an optimal value is no greater than a suboptimal one.}

\textcolor{black}{Then we have iii) $\bar{I}^{*,I} \leq \bar{A}^{*,I}$ and iv) $\bar{I}^{*,A} \leq \bar{A}^{*,A}$, since $I(k) \leq A(k) \ \forall k$\footnote{This holds since the system is action-instantaneous, meaning data is either actuated or dropped upon reception. If there is a queue within the receiver, it does not hold. For a detailed discussion, refer to \cite{nikkhah2026semantics}.}, and then the same holds for their averages. Combining i), iv), and ii), we obtain (\ref{equ:relation_4_values}).}
\end{proof}

\begin{figure*}[t]
    \centering
    \begin{minipage}{0.08\textwidth}
        \centering
        \includegraphics[width=\linewidth]{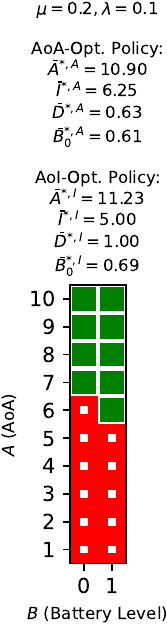}
    \end{minipage}\hfill
    \begin{minipage}{0.08\textwidth}
        \centering
        \includegraphics[width=\linewidth]{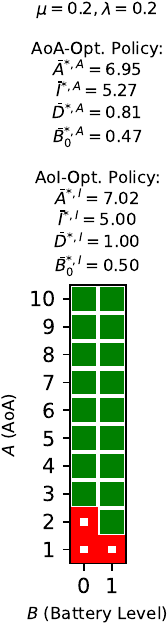}
    \end{minipage}\hfill
    \begin{minipage}{0.08\textwidth}
        \centering
        \includegraphics[width=\linewidth]{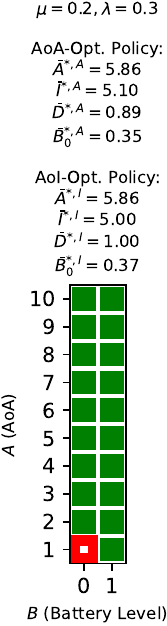}
    \end{minipage}\hfill
    \begin{minipage}{0.08\textwidth}
        \centering
        \includegraphics[width=\linewidth]{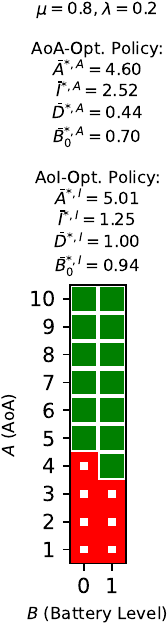}
    \end{minipage}\hfill
    \begin{minipage}{0.08\textwidth}
        \centering
        \includegraphics[width=\linewidth]{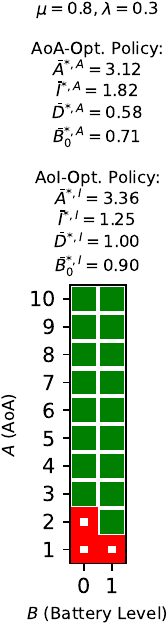}
    \end{minipage}\hfill
    \begin{minipage}{0.08\textwidth}
        \centering
        \includegraphics[width=\linewidth]{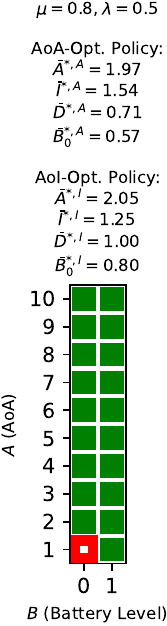}
    \end{minipage}\hfill
    \begin{minipage}{0.24\textwidth}
        \centering
        \includegraphics[width=\linewidth]{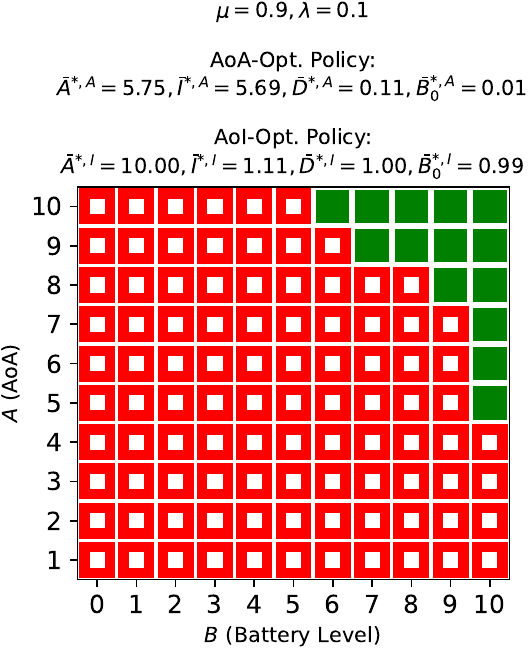}
    \end{minipage}\hfill
    \begin{minipage}{0.24\textwidth}
        \centering
        \includegraphics[width=\linewidth]{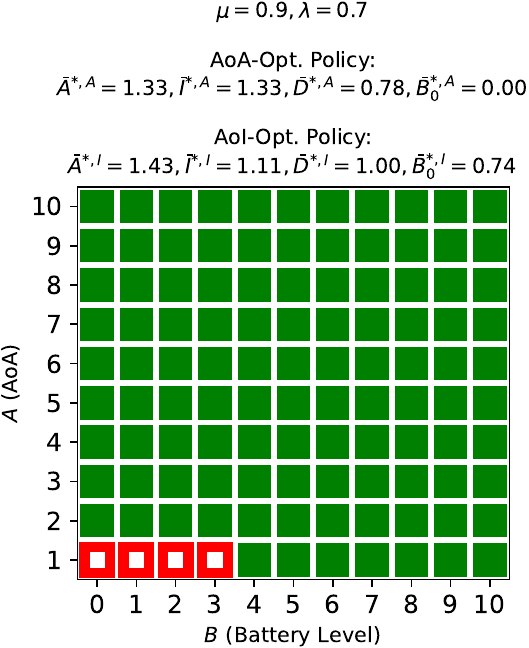}
    \end{minipage}

    \caption{The AoA-Opt. policies and the values of the average AoA, the average AoI, the transmission rate, and the probability of an empty battery under both AoA-Opt. and AoI-Opt. \textcolor{black}{policies} for different values of $\mu$ and $\lambda$, with $B_m=1$ (left side) and $B_{m}=10$ (right side).}
    \label{fig:Optimal_Policies}
\end{figure*}

\begin{figure*}[h!]
\centering

\begin{tikzpicture}

\node[inner sep=0pt] (fig) {
\begin{minipage}{\textwidth}
\centering

\begin{minipage}{0.24\textwidth}
\centering
\includegraphics[width=\linewidth]{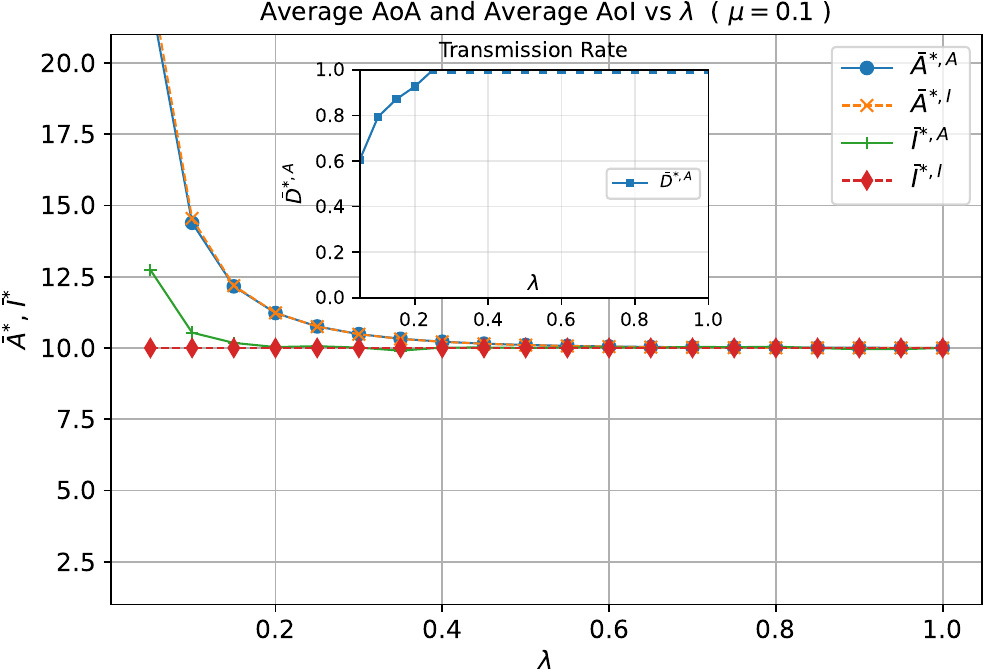}
\end{minipage}\hfill
\begin{minipage}{0.24\textwidth}
\centering
\includegraphics[width=\linewidth]{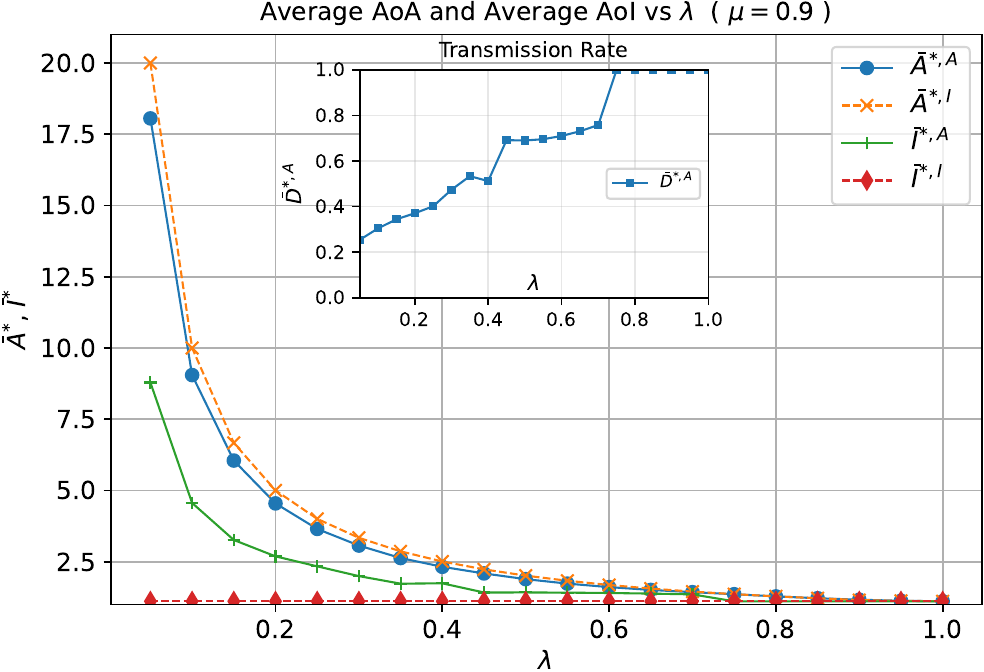}
\end{minipage}\hfill
\begin{minipage}{0.24\textwidth}
\centering
\includegraphics[width=\linewidth]{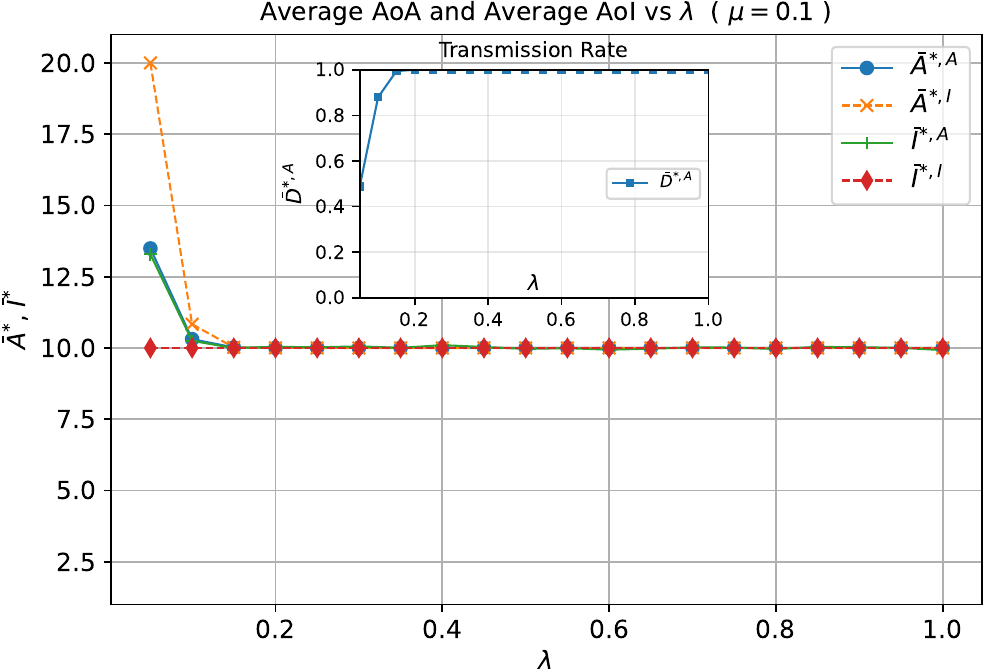}
\end{minipage}\hfill
\begin{minipage}{0.24\textwidth}
\centering
\includegraphics[width=\linewidth]{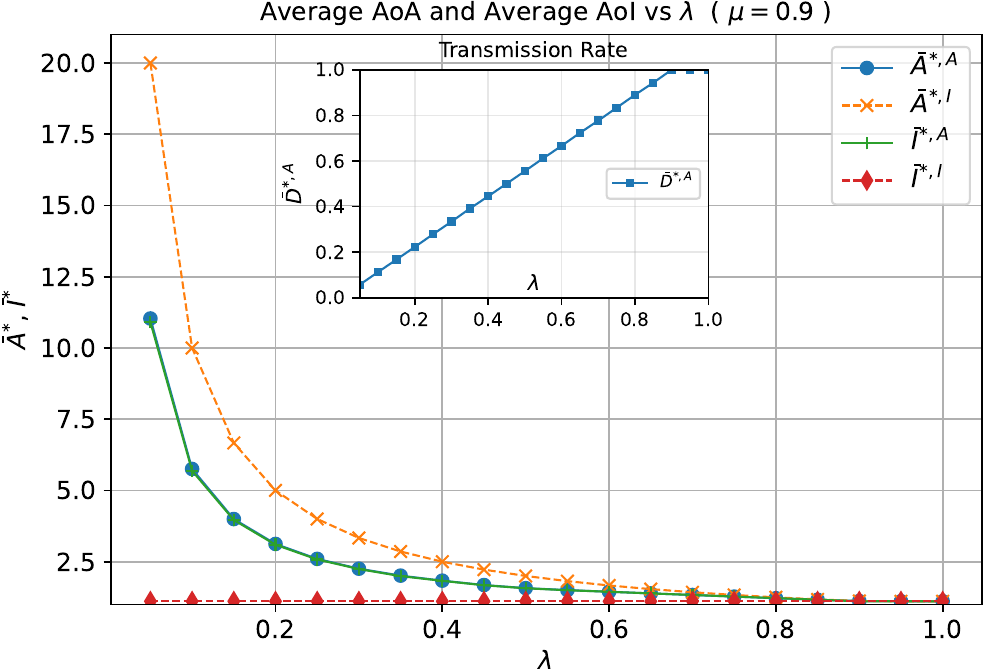}
\end{minipage}

\vspace{2mm}

\begin{minipage}{0.24\textwidth}
\centering
\includegraphics[width=\linewidth]{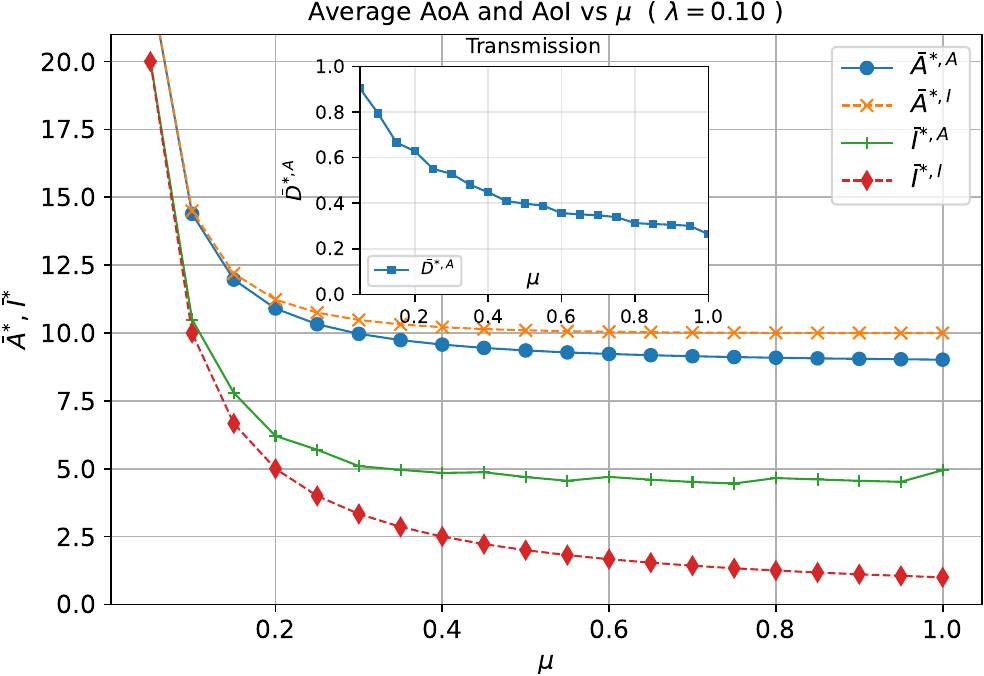}
\end{minipage}\hfill
\begin{minipage}{0.24\textwidth}
\centering
\includegraphics[width=\linewidth]{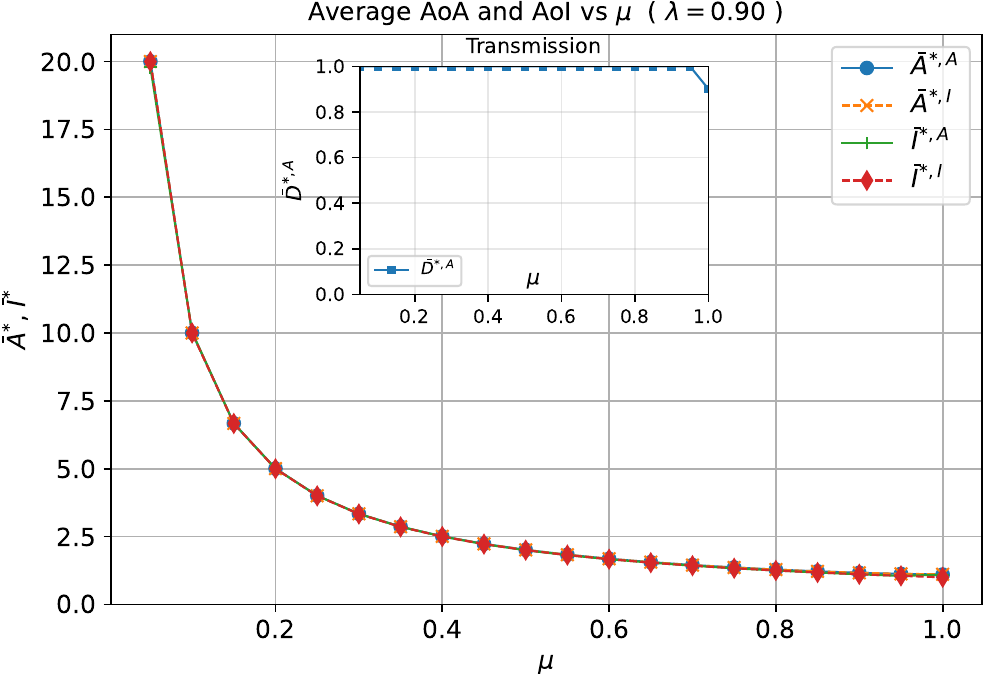}
\end{minipage}\hfill
\begin{minipage}{0.24\textwidth}
\centering
\includegraphics[width=\linewidth]{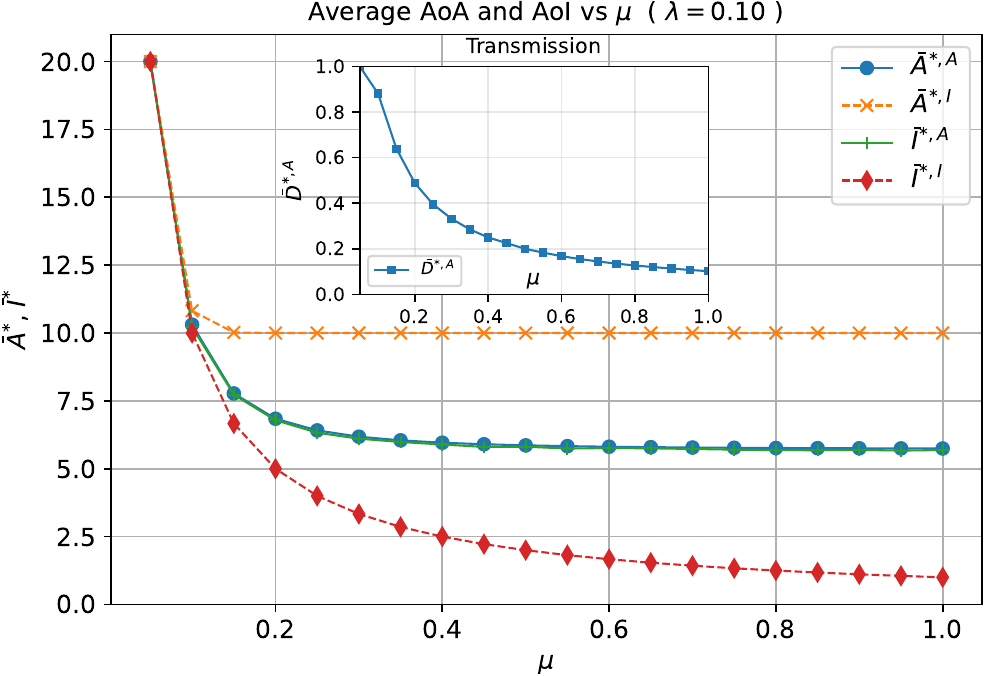}
\end{minipage}\hfill
\begin{minipage}{0.24\textwidth}
\centering
\includegraphics[width=\linewidth]{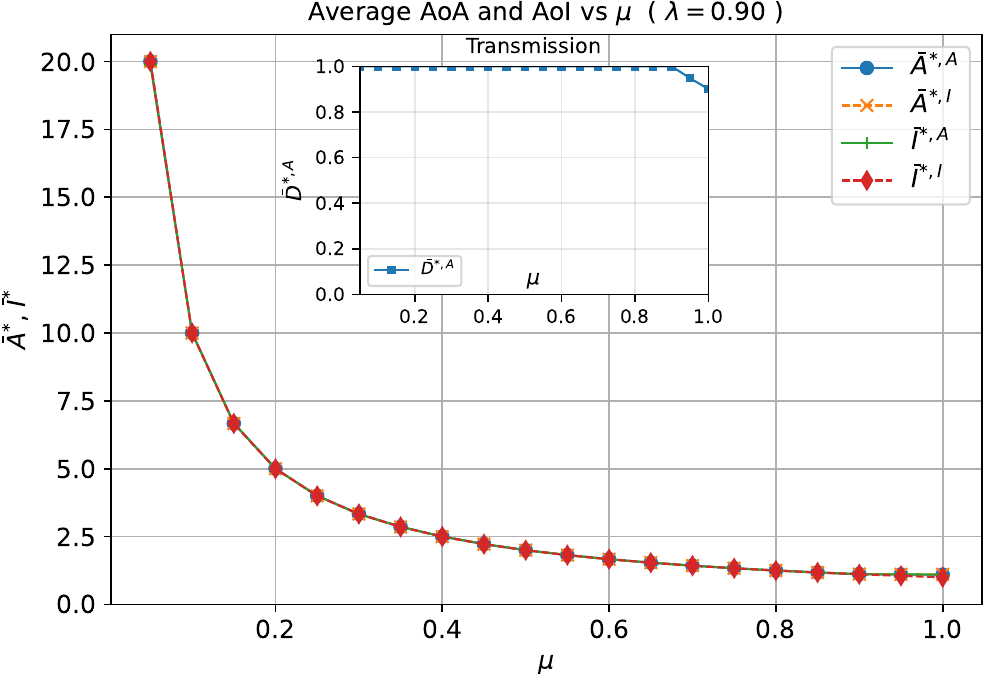}
\end{minipage}

\end{minipage}
};

\draw[gray, dotted, line width=1pt]
  ($(fig.north) + (0,0)$)
  --
  ($(fig.south) + (0,0)$);

\draw[gray, dotted, line width=0.6pt]
  ($(fig.west) + (0,0cm)$)
  --
  ($(fig.east) + (0,0cm)$);

\end{tikzpicture}

\caption{The average AoA and the average AoI under the AoA-Opt.\ and AoI-Opt.\ policies versus $\lambda$ for two values of $\mu$ (upper side), and versus $\mu$ for two values of $\lambda$ (lower side), with $B_m=1$ (left side) and $B_m=10$ (right side), along with the corresponding transmission rate under the AoA-Opt.\ policy in the inset figures.}
\label{fig:Compare_MDP_MD_AoA_and_AoI_l12}
\end{figure*}

\begin{figure}
    \centering
    \includegraphics[width=\linewidth]{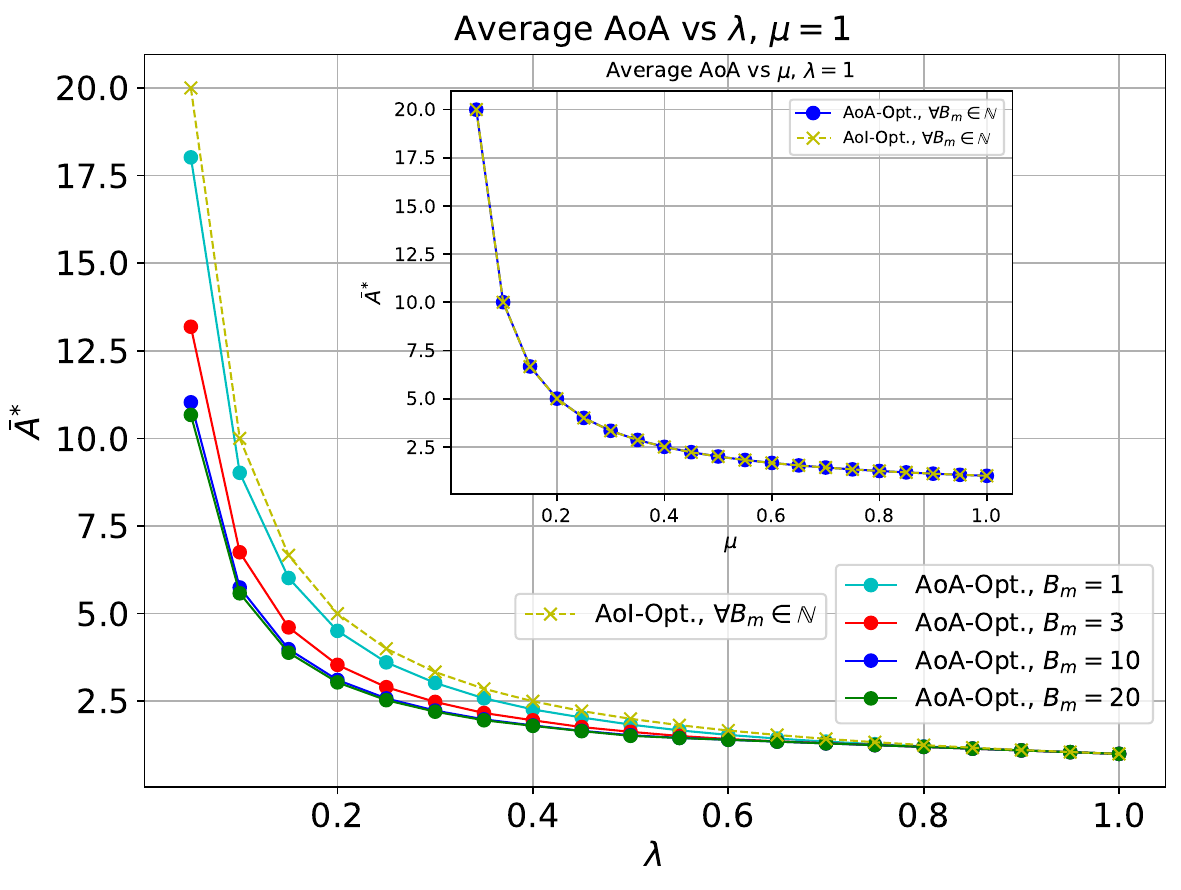}
    \caption{\textcolor{black}{The average AoA under the AoA-Opt.\ and AoI-Opt.\ policies versus $\lambda$ for $\mu=1$ (main figure) and versus $\mu$ for $\lambda=1$ (inset), for different battery capacities.}}
    \label{fig:Compare_MDP_MP_AoA_main_B_1_3_10_20_with_inset_B_1_A_m_150}
\end{figure}

\section{Results} \label{Results}
In this section, we present and compare the average AoA and average AoI under the AoA-Opt. and AoI-Opt. policies, while also considering the transmission rate as a potential cost, and provide interpretation and analysis.

In the AoA-Opt. policy, the optimal average AoA, $\bar{A}^{*,A}$, is obtained by solving the MDP described in Section \ref{MDP Characterization} using the RVI algorithm (Algorithm \ref{alg:RVI}), and validated through simulations. The RVI is \textcolor{black}{run} with a convergence threshold of $\epsilon=0.01$ and a truncation level of $A_m=300$. In addition, the suboptimal average AoI, i.e., $\bar{I}^{*,A}$, transmission rate (\ref{equ:active_fraction}), and probability of an empty battery (\ref{equ:empty_battery}) are computed by simulating the system under the obtained optimal policy for each pair of $\mu$ and $\lambda$. More details are in Section \ref{Age of Actuation (AoA) Optimization}.

For the AoI-Opt. policy, the optimal average AoI is simply $\bar{I}^{*,I}=\frac{1}{\mu}$, which can also be confirmed by solving a trivial separate MDP. The suboptimal average AoA is derived from \textcolor{black}{(\ref{equ:Average_AoA_MC_body_general}) and (\ref{equ:Average_AoA_MC_body_equal})}. The transmission rate is $\bar{D}^{*,I}=1$, and the probability of an empty battery is $\bar{B}_{0}^{*,I}=\frac{1-r}{1-r^{B_{m}+1}}$ \textcolor{black}{for $\mu \neq \lambda$, and $1/(B_{m}+1)$ for $\mu=\lambda$}, all as discussed in Section \ref{Age of Information (AoI) Optimization}.

\subsection{Illustration of the Results} \label{Illustration of the Results}
Fig.~\ref{fig:Optimal_Policies} illustrates the AoA-Opt. policy (i.e., the action taken at each state), along with the average AoA, average AoI, the transmission rate, and the probability of an empty battery for both the AoA-Opt. policy and AoI-Opt. policy, for two values of battery capacity: $B_{m}=1$ on the left and $B_{m}=10$ on the right. For $B_{m}=1$, results are shown for six pairs of $\mu$ and $\lambda$, while for $B_{m}=10$, results are shown for two pairs. The optimal action at each state is represented as follows: solid green squares indicate transmission ($d=1$), while red-outlined (hollow) squares indicate no transmission ($d=0$).

Fig.~\ref{fig:Compare_MDP_MD_AoA_and_AoI_l12} compares the average AoA and average AoI under both the AoA-Opt.\ policy and the AoI-Opt.\ policy. The results are organized into two groups: $B_{m}=1$ (left) and $B_{m}=10$ (right). Each group is further divided into two subgroups: one showing results for two values of $\mu$ versus $\lambda$ (top), and one for two values of $\lambda$ versus $\mu$ (bottom).

Fig.~\ref{fig:Compare_MDP_MP_AoA_main_B_1_3_10_20_with_inset_B_1_A_m_150} shows the average AoA under both the AoA-Opt.\ and AoI-Opt.\ policies for different values of $B_{m}$, with the main figure depicting results for $\mu=1$ versus $\lambda$, and the inset figure depicting results for $\lambda=1$ versus $\mu$.

\begin{figure}[t]
    \centering

        \begin{subfigure}[b]{0.49\linewidth}
        \centering
        \includegraphics[width=\linewidth]{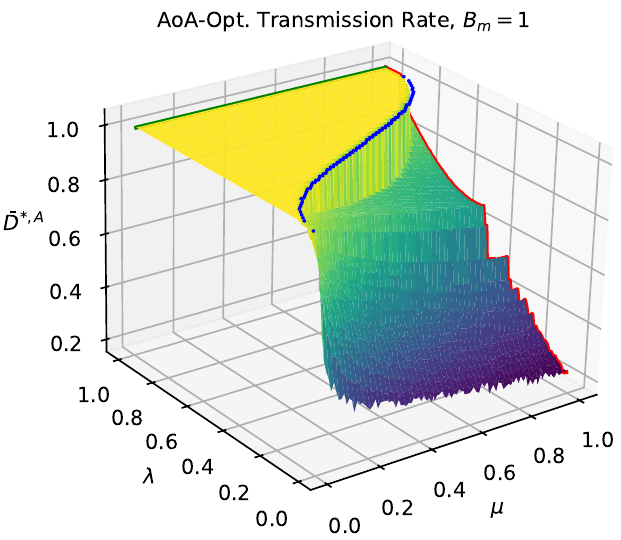}
        \caption{$B_m = 1$}
        \label{fig:Active_Fraction_Bm1}
    \end{subfigure}
    \hfill
    \begin{subfigure}[b]{0.49\linewidth}
        \centering
        \includegraphics[width=\linewidth]{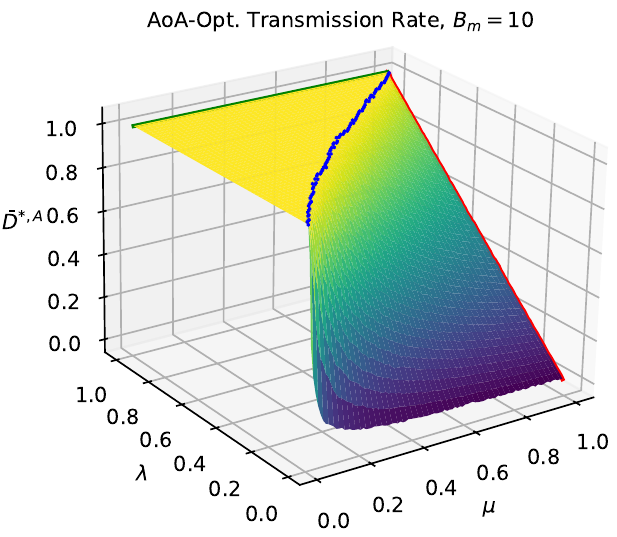}
        \caption{$B_m = 10$}
        \label{fig:Active_Fraction_Bm10}
    \end{subfigure}

    \caption{The transmission rate under AoA-Opt. policy versus $\mu$ and $\lambda$ for two values of $B_m$.}
    \label{fig:Active_Fraction}
\end{figure}

\begin{figure}
    \centering
    \includegraphics[width=\linewidth]{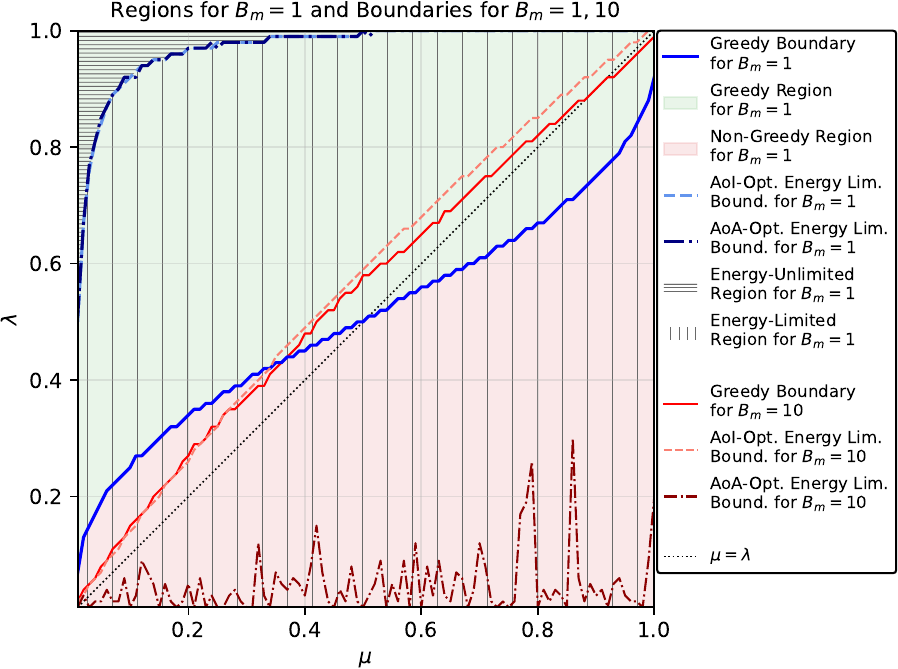}
    \caption{Greedy and non-greedy regions and energy-limited and energy-unlimited regions for $B_{m}=1$ and the corresponding boundaries for $B_{m}=1$ and $B_{m}=10$.}
    \label{fig:Cost_Regions}
\end{figure}

\begin{figure}
    \centering
    \includegraphics[width=0.9\linewidth]{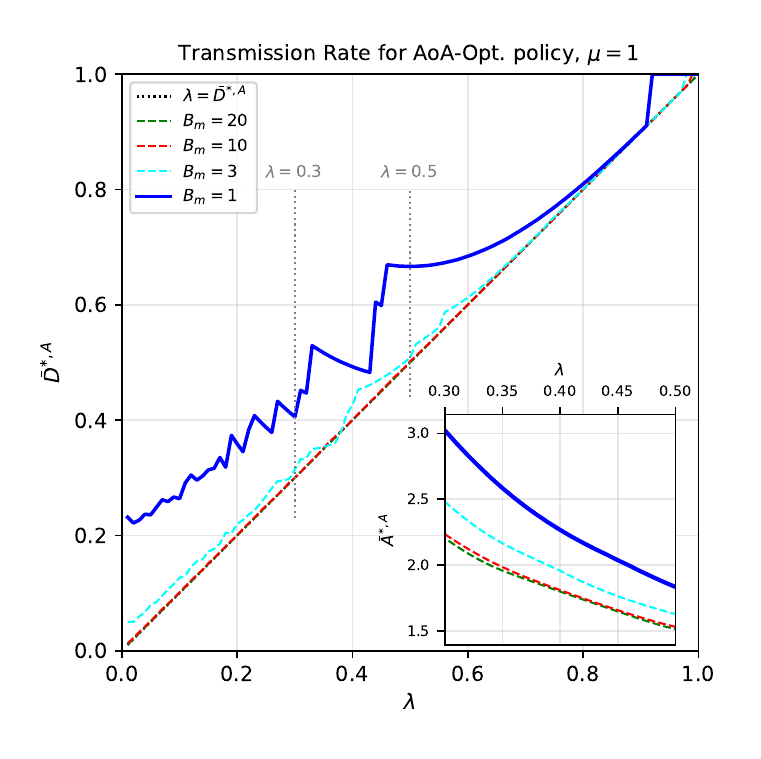}
    \caption{Transmission rate for $\mu=1$ for different values of $B_{m}$.}
    \label{fig:Pr_Trans_Only_l1_1}
\end{figure}

Fig.~\ref{fig:Active_Fraction} presents the transmission rate~(\ref{equ:active_fraction}) obtained via simulation under the AoA-Opt.\ policy, plotted versus $\mu$ and $\lambda$ (each with step size $0.01$), over $10^{5}$ time slots for each of the $10^{4}$ $(\mu, \lambda)$ pairs, for $B_{m}=1$ and $B_{m}=10$. The blue curve separates the region where $\bar{D}^{*,A} \approx 1$ from the region where $\bar{D}^{*,A} < 1$. The green curve corresponds to $\lambda=1$ and the red curve corresponds to $\mu=1$.

Fig.~\ref{fig:Cost_Regions} shows three boundaries for each of $B_{m}=1$ and $B_{m}=10$: i) the boundary separating the regions $\bar{D}^{*,A} \approx 1$ and $\bar{D}^{*,A} < 1$ (the blue curve in Fig.~\ref{fig:Active_Fraction}), ii) the boundary separating regions $\bar{B}_{0}^{*,A} \approx 0$ and $\bar{B}_{0}^{*,A} > 0$, and iii) the boundary separating regions $\bar{B}_{0}^{*,I} \approx 0$ and $\bar{B}_{0}^{*,I} > 0$, along with the two resulting pairs of separated regions for i) and for ii) and iii) together\footnote{The boundaries coincide for $B_{m}=1$.}. The data for i) and ii) are extracted from the same data as Fig.~\ref{fig:Active_Fraction}, and iii) is computed numerically from the expression of $\bar{B}_{0}^{*,I}$ discussed in Section \ref{Age of Information (AoI) Optimization}.

Fig.~\ref{fig:Pr_Trans_Only_l1_1} shows the transmission rate for $\mu=1$ under the AoA-Opt.\ policy (corresponding to the red curve in Fig.~\ref{fig:Active_Fraction}) for different values of $B_{m}$, obtained by running $10^7$ time slots, yielding a more accurate illustration compared to Figs.~\ref{fig:Active_Fraction} and~\ref{fig:Cost_Regions}. The main figure also contains an inset showing the average AoA under the AoA-Opt.\ policy for $\mu=1$ and for those different values of $B_{m}$, over the range $\lambda \in [0.3, 0.5]$.

\subsection{Analysis of the Results} \label{Analysis of the Results}
In this section, we analyze and interpret the results in terms of timeliness and cost.
Timeliness is measured by the average AoA and average AoI, while cost refers to the transmission rate required for sensing and transmission.
This cost is not explicitly modeled in the problem, as it can be assumed to be negligible compared to the energy consumption associated with actuation and/or of a fundamentally different nature, and is therefore not directly comparable to the actuation energy cost.

Figs.~\ref{fig:Optimal_Policies}, \ref{fig:Compare_MDP_MD_AoA_and_AoI_l12}, and~\ref{fig:Compare_MDP_MP_AoA_main_B_1_3_10_20_with_inset_B_1_A_m_150} show that the average AoA under the AoA-Opt.\ policy is \textcolor{black}{lower than or equal to that} under the AoI-Opt.\ policy, and conversely, the average AoI under the AoI-Opt.\ policy is \textcolor{black}{lower than or equal to that} under the AoA-Opt.\ policy.
\textcolor{black}{Also, the relation among the four corresponding curves in the plots of Fig. \ref{fig:Compare_MDP_MD_AoA_and_AoI_l12} is consistent with (\ref{equ:relation_4_values}).}

\textbf{Greedy and Non-Greedy Regions} are where the AoA-Opt. policy tends to the AoI-Opt. policy and where it does not, respectively. As the system goes into the greedy\footnote{Since in this region the AoA-Opt. policy sets the transmitter to \textit{greedily} transmit at all states, and consequently at all time slots (which is equivalent to the AoI-Opt. policy), we call this the \textit{greedy} region.} region ($\bar{D}^{*,A} \rightarrow 1$, as seen from the dominance of green actions in Fig.~\ref{fig:Optimal_Policies} and the surface approaching $1$ in Fig.~\ref{fig:Active_Fraction}), the average AoA values under the two policies coincide with each other ($\bar{A}^{*,A} \rightarrow \bar{A}^{*,I}$, the blue and orange curves coinciding in Fig. \ref{fig:Compare_MDP_MD_AoA_and_AoI_l12}), and so do the average AoI values ($\bar{I}^{*,A} \rightarrow \bar{I}^{*,I}$, the green and red curves coinciding in Fig. \ref{fig:Compare_MDP_MD_AoA_and_AoI_l12}). This corresponds to the green region in Fig.~\ref{fig:Cost_Regions}, which is separated from the orange region by the blue boundary for $B_{m}=1$. Notably, this generally occurs for relatively higher values of $\lambda$ compared to $\mu$, as can be seen in Figs.~\ref{fig:Active_Fraction} and~\ref{fig:Cost_Regions}. \textcolor{black}{We observe that for $B_{m}=10$, the boundary is already close to the line $\mu=\lambda$.} However, for $B_{m}=1$, there is an interesting symmetry around the point $\mu=\lambda=0.5$, as seen in Fig.~\ref{fig:Cost_Regions}.

\textbf{Energy-unlimited and Energy-Limited Regions} are where the values of AoA and AoI themselves coincide and where they do not, respectively. An energy-unlimited region is where the probability of an empty battery asymptotically approaches zero\footnote{The probability of an empty battery never reaches absolute zero for a finite-sized battery.}, either under the AoI-Opt. policy ($\bar{A}^{*,I} \rightarrow \bar{I}^{*,I}$, the orange and red curves coinciding in Fig. \ref{fig:Compare_MDP_MD_AoA_and_AoI_l12}) or the AoA-Opt. policy ($\bar{A}^{*,A} \rightarrow \bar{I}^{*,A}$, the blue and green curves coinciding in Fig. \ref{fig:Compare_MDP_MD_AoA_and_AoI_l12}).
\textcolor{black}{By (\ref{equ:relation_4_values}), the former implies the latter, and Fig.~\ref{fig:Compare_MDP_MD_AoA_and_AoI_l12} shows that the converse does not hold.} Fig. \ref{fig:Cost_Regions} shows that for $B_{m}=1$ these two energy limitation regions coincide, while they do not for $B_{m}=10$, and as noted above, the energy-unlimited region of the AoA-Opt. policy contains the energy-unlimited region of the AoI-Opt. policy.

Fig.~\ref{fig:Compare_MDP_MP_AoA_main_B_1_3_10_20_with_inset_B_1_A_m_150} shows that when $\mu=1$, the average AoA is indifferent to the battery capacity under the AoI-Opt.\ policy, whereas under the AoA-Opt.\ policy it is lower and decreases further as the battery capacity increases. The inset shows that when $\lambda=1$, the AoA reduces to the AoI and their averages are equal ($\frac{1}{\mu}$), both independent of the battery capacity.

Also, Fig.~\ref{fig:Compare_MDP_MD_AoA_and_AoI_l12} shows that when the battery capacity is larger, such as $B_{m}=10$, under the AoA-Opt.\ policy, the average AoA is very close to the average AoI. This is due to the fact that, to keep the AoA low, \textcolor{black}{the AoA-Opt.\ policy maintains sufficient energy for actuation and avoids emptying the battery}; consequently, very low values of $\bar{B}_{0}^{*,A}$ are observed in Fig.~\ref{fig:Optimal_Policies} for $B_{m}=10$, compared to the very high values of $\bar{B}_{0}^{*,I}$. This is also evident in Fig. \ref{fig:Cost_Regions}, as for $B_{m}=10$, the boundaries of energy limitation (above the boundaries are the energy-unlimited regions) are totally different, and most of the plane is energy-unlimited under AoA-Opt. policy. However, for lower battery capacities, such as $B_{m}=1$, maintaining a non-empty battery becomes harder, as again observed in Fig.~\ref{fig:Optimal_Policies}, and the average AoA under the AoA-Opt.\ policy tends toward that under the AoI-Opt.\ policy, since \textcolor{black}{a unit battery leaves little room to store energy for actuation}. Fig. \ref{fig:Cost_Regions} confirms this for $B_{m}=1$ by showing that the boundaries of energy limitation for the two policies coincide.

Fig.~\ref{fig:Pr_Trans_Only_l1_1} illustrates how the transmission rate self-regulates based on the energy arrival rate when optimizing the average AoA, analyzed here for a channel success rate of $\mu=1$. \textcolor{black}{We observe that as $B_{m}$ increases, $\bar{D}^{*,A}$ approaches $\lambda$}, meaning the transmission rate approaches the energy arrival rate as the battery capacity is relaxed. However, for $B_{m}=1$, the behavior is highly non-linear. First, the transmission rate is generally above the line $\lambda=\bar{D}^{*,A}$, indicating that a low battery capacity requires a higher transmission rate to compensate, unlike the case for higher battery capacities. Second, there is a counterintuitive non-monotonicity: for example, over the range $\lambda \in [0.3, 0.5]$, as $\lambda$ increases linearly and the optimal average AoA shown in the inset decreases monotonically, the corresponding transmission rate first rises, then drops, then rises sharply again before saturating.

\section{Conclusion}
We studied transmission scheduling for an energy-harvesting receiver whose actuations require both data and energy. We formulated average-AoA minimization as a Markov decision process and characterized an AoA-threshold optimal policy. We also derived the average AoA in closed form under the AoI-optimal policy, which is always-transmit. Comparing these policies shows that optimizing information freshness alone need not optimize actuation freshness, and that the two policies coincide only for some pairs of channel success and energy-arrival probabilities. State-dependent scheduling can improve average AoA and reduce transmission attempts, although it may increase the average AoI. The numerical results also reveal non-monotonic transmission rates under the AoA-optimal policy for a one-unit battery, illustrating that improved actuation freshness does not require a monotonically increasing transmission rate. An analytical characterization of the observed battery-threshold structure remains an open question.

\begin{appendices}
\section{Proof of Theorem~\ref{theo:WAUniErg}} \label{App:WAUniErg}
\noindent
\textcolor{black}{Here, we successively prove the elements of Theorem~\ref{theo:WAUniErg}.}
\subsection{Proof of Theorem~\ref{theo:WAUniErg}\ref{theo:WAUniErg_unichain}} \label{App:WAUniErg_unichain}
A unichain policy is one whose corresponding induced Markov chain contains a single recurrent class (and possibly some transient states)\cite[Section 5.2]{bertsekas2012dynamic}.
We now show that every stationary policy of the presented MDP is unichain.

The state $(A_m,B_m)$ is accessible from any state under any policy. We divide the set of all $|\Pi|=2^{A_m \times \textcolor{black}{(B_{m}+1)}}$ policies into two subsets: those with $d_{(A_m,B_m)}=0$ and those with $d_{(A_m,B_m)}=1$. In the first subset where $d_{(A_m,B_m)}=0$, the recurrent class contains only the state $(A_m,B_m)$, since it absorbs the chain upon entry. In the second subset where $d_{(A_m,B_m)}=1$, all states $(A,B)$ where $A\in \{1,\cdots,A_m\}$ and $B \in \{B_m-1,B_m\}$ are in the recurrent class. If at least one state $(A,B)$ where $A\in \{1,\cdots,A_{m}\}$ and $B=B_{m}-1$ has $d_{(A,B)}=1$, then all states $(A,B)$ where $A\in \{1,\cdots,A_{m}\}$ and $B=B_{m}-2$ are in the recurrent class as well. Continuing this argument, if at least one state $(A,B)$ where $A\in \{1,\cdots,A_{m}\}$ and $B=B_{m}-2$ has $d_{(A,B)}=1$, then all states $(A,B)$ where $A\in \{1,\cdots,A_{m}\}$ and $B=B_{m}-3$ are also in the recurrent class, and so forth down to $B=0$. Note that this connection must be consistent: for example, if there is at least one state with $d=1$ for each group from $B=B_m-1$ down to $B=\hat{B}+1$, but no state with $d=1$ among states with $B=\hat{B}$, and again at least one state with $d=1$ among states with $B=\hat{B}-1$, then the recurrent class contains only states $(A,B)$ where $A\in \{1,\cdots,A_{m}\}$ and $B\in \{\hat{B},\cdots,B_m\}$, and not states with $B=\hat{B}-1$, $B=\hat{B}-2$, or below. We can describe the set of states in the recurrent class based on the policy by
\begin{equation}
\begin{aligned}
&\mathbb{RC}=\\
&\begin{cases}
(A_{m},B_{m})  &  d_{(A_{m},B_{m})}=0,\\[2ex]
\begin{aligned}
&(\{A\in \{1,\cdots,A_{m}\},\\
&B\in \{\hat{B},\cdots,B_m\}\})
\end{aligned}
&
\begin{aligned}
&(d_{(A_{m},B_{m})}=1) \land \\
&(\hat{B}=\min\{\textcolor{black}{0},\cdots,\textcolor{black}{B_{m}-1}\}\\
&| \ \{ \forall B=\{\hat{B}+1,\cdots,B_{m}\}\\
& \exists A \ | \ d_{(A,B)}=1\}).\\
\end{aligned}
\end{cases}
\end{aligned}
\end{equation}
As shown above, every policy induces a single recurrent class, and is therefore unichain. Thus, all stationary policies of the presented MDP are unichain. \hfill \qed

\subsection{Proof of Theorem~\ref{theo:WAUniErg}\ref{theo:WAUniErg_WA}} \label{App:WAUniErg_WA}
For an MDP, if all stationary policies are unichain, then the weak accessibility condition holds \cite[Proposition 5.2.5]{bertsekas2012dynamic}. By \textcolor{black}{Theorem~\ref{theo:WAUniErg}\ref{theo:WAUniErg_unichain}}, all stationary policies of the presented MDP are unichain, and therefore the MDP is weakly accessible. \hfill \qed

\subsection{Proof of Theorem~\ref{theo:WAUniErg}\ref{theo:WAUniErg_Erg}} \label{App:WAUniErg_Erg}
Since all policies are unichain, the optimal ones are also unichain \cite[Proposition 5.2.6]{bertsekas2012dynamic}, and their induced Markov chains contain exactly one recurrent class with possibly some transient states. This alone guarantees the existence of a unique stationary distribution for each induced Markov chain. Furthermore, since the state $(A_m,B_m)$ belongs to the recurrent class of all induced Markov chains and can transition to itself in one step in all of them \textcolor{black}{(w.p. $1$ under $d=0$ and w.p. $\bar{\mu}>0$ under $d=1$, since $\mu<1$)}, its period is $1$, and it is therefore aperiodic. Since aperiodicity is a class property, if one state in a finite recurrent class is aperiodic, then all states in the class are aperiodic, and the class itself is ergodic. Thus, the recurrent class is ergodic and the induced Markov chain converges to its unique stationary distribution\footnote{The stationary probability of any transient state is $0$.} \cite{norris1997markov}. \hfill \qed

\subsection{Proof of Theorem~\ref{theo:WAUniErg}\ref{theo:WAUniErg_Indep}} \label{App:WAUniErg_Indep}
\textcolor{black}{By Theorem~\ref{theo:WAUniErg}\ref{theo:WAUniErg_WA}, the presented MDP is weakly accessible, and therefore the optimal average cost is the same for all initial states \cite[Proposition 5.2.3]{bertsekas2012dynamic}. By Theorem~\ref{theo:WAUniErg}\ref{theo:WAUniErg_unichain}, all stationary policies are unichain, and by Theorem~\ref{theo:WAUniErg}\ref{theo:WAUniErg_Erg}, their recurrent classes are aperiodic. Under these conditions, the RVI algorithm converges \cite[Theorem 8.5.4 and Section 8.5.5]{puterman1994markov}. Therefore, Algorithm~\ref{alg:RVI} converges and yields the optimal average cost (average AoA) in (\ref{equ:Average_AoA_MDP}), independently of the starting state.} \hfill \qed

\subsection{Proof of Theorem~\ref{theo:WAUniErg}\ref{theo:WAUniErg_Irre}} \label{App:WAUniErg_Irre}
By Theorem~\ref{theo:threshold_A}, for each battery level $B$, the optimal policy has a threshold structure in age $A$: once $d_{(A,B)}=1$ is observed for some $A$, all subsequent states with higher $A$ also have $d=1$. \textcolor{black}{We assume that $A_m$ is large enough that, for each $B$, the optimal policy has at least one state with $d=1$. For $A_m=300$, this holds for all optimal policies obtained in Section~\ref{Results} for $0<\mu,\lambda<1$, including the lowest battery levels and the lowest values of $\mu$ and $\lambda$.} Therefore, by the argument in Appendix \ref{App:WAUniErg_unichain}, all states are captured in the recurrent class, and the induced Markov chains from the optimal policies have no transient states, making them fully irreducible. This strengthens the results of \textcolor{black}{Theorems~\ref{theo:WAUniErg}\ref{theo:WAUniErg_Erg} and \ref{theo:WAUniErg}\ref{theo:WAUniErg_Exist}}, for optimal policies.

\subsection{Proof of Theorem~\ref{theo:WAUniErg}\ref{theo:WAUniErg_Exist}} \label{App:WAUniErg_Exist}
By \textcolor{black}{Theorem~\ref{theo:WAUniErg}\ref{theo:WAUniErg_Erg}}, the induced Markov chains from all policies, including the optimal ones, are ergodic. Therefore, the stationary distributions exist and are unique, and consequently the probabilities (\ref{equ:active_fraction}) and (\ref{equ:empty_battery}) exist and are independent of the starting state. \hfill \qed

\section{Proof of Lemma \ref{lemma:increasing}} \label{App:lemma:increasing}

\subsection{Proof of Lemma \ref{lemma:increasing_A}} \label{App:lemma:increasing_A}
By mathematical induction, if we show (i) the property holds for the first iteration of the RVI, i.e., $k=0$, and (ii) if it holds for the $k$-th iteration (\ref{equ:increasing_A}), then it also holds for the $(k+1)$-th iteration, it follows that the property holds for all iterations, including the converged iteration corresponding to the optimal cost function.

To verify (i), since we initialize such that $h(s), h_{\text{new}}(s) \leftarrow 0, \ \forall s \in S$, it holds. To verify (ii), let $s_{1}\equiv s_{(A_1,B)}$ and $s_{2} \equiv s_{(A_2,B)}$ be two states where $A_1 \leq A_2$. Then, we define

\begin{equation} \label{alpha}
\alpha_{i,j} = A_{i} + \sum_{s' \in S} p_{s_{i},d=j,s_{i}'} h^{k}(s_{i}'), \ i \in \{1,2\}, \ j \in \{0,1\}. 
\end{equation}

In each iteration of the RVI, we minimize over the two actions, $d=0$ and $d=1$, to obtain the updated cost function for each state. Based on the definition in (\ref{alpha}), the updates for $s_1$ and $s_2$ are given by
\begin{equation} \label{equ:min_hs_A}
h_{s_{1}}^{k+1}=\min\{\alpha_{1,0},\alpha_{1,1}\}, \
h_{s_{2}}^{k+1}=\min\{\alpha_{2,0},\alpha_{2,1}\},
\end{equation} 
respectively.

We compare the update terms in (\ref{alpha}) with $j=0$ (no transmission) together and with $j=1$ (transmission) together, across three cases of the battery.
\textcolor{black}{Based on (\ref{equ:BellmanEqu_B_0}), (\ref{equ:BellmanEqu_B_0_B_B_m}), and (\ref{equ:BellmanEqu_B_m})}, for $\boldsymbol{B=0}$
\begin{subequations} \label{For_A_1}
\begin{align*}
\alpha_{1,0}=A_1+&\lambda h_{(\textcolor{black}{A_{1}^{\text{tr}}},1)} + \bar{\lambda} h_{(\textcolor{black}{A_{1}^{\text{tr}}},0)},\numberthis\label{For_A_1_10}\\
\alpha_{2,0}=A_2+&\lambda h_{(\textcolor{black}{A_{2}^{\text{tr}}},1)} + \bar{\lambda} h_{(\textcolor{black}{A_{2}^{\text{tr}}},0)},\numberthis\label{For_A_1_20}\\
\alpha_{1,1}=A_1+&w h_{(1,0)} + \bar{\lambda} h_{(\textcolor{black}{A_{1}^{\text{tr}}},0)}+ y h_{(\textcolor{black}{A_{1}^{\text{tr}}},1)},\numberthis\label{For_A_1_11}\\
\alpha_{2,1}=A_2+&w h_{(1,0)} + \bar{\lambda} h_{(\textcolor{black}{A_{2}^{\text{tr}}},0)} + y h_{(\textcolor{black}{A_{2}^{\text{tr}}},1)},\label{For_A_1_21}\numberthis
\end{align*}
\end{subequations}
for $\boldsymbol{0 < B < B_m}$
\begin{subequations} \label{For_A_2}
\begin{align*}
\alpha_{1,0}=A_1+&\lambda h_{(\textcolor{black}{A_{1}^{\text{tr}}},B+1)} + \bar{\lambda} h_{(\textcolor{black}{A_{1}^{\text{tr}}},B)},\label{For_A_2_10}\numberthis\\
\alpha_{2,0}=A_2+&\lambda h_{(\textcolor{black}{A_{2}^{\text{tr}}},B+1)} + \bar{\lambda} h_{(\textcolor{black}{A_{2}^{\text{tr}}},B)},\label{For_A_2_20}\numberthis\\
\alpha_{1,1}=A_1+&w h_{(1,B)} + x h_{(1,B-1)} + z h_{(\textcolor{black}{A_{1}^{\text{tr}}},B)}\\
+& y h_{(\textcolor{black}{A_{1}^{\text{tr}}},B+1)},\label{For_A_2_11}\numberthis\\
\alpha_{2,1}=A_2+&w h_{(1,B)} + x h_{(1,B-1)} + z h_{(\textcolor{black}{A_{2}^{\text{tr}}},B)}\\
+& y h_{(\textcolor{black}{A_{2}^{\text{tr}}},B+1)},\label{For_A_2_21} \numberthis
\end{align*}
\end{subequations}
and for $\boldsymbol{B=B_m}$
\begin{subequations}\label{For_A_3}
\begin{align*}
\alpha_{1,0}=A_1+& h_{(\textcolor{black}{A_{1}^{\text{tr}}},B_m)},\label{For_A_3_10}\numberthis\\
\alpha_{2,0}=A_2+& h_{(\textcolor{black}{A_{2}^{\text{tr}}},B_m)},\label{For_A_3_20}\numberthis\\
\alpha_{1,1}=A_1+&w h_{(1,B_m)} + x h_{(1,B_m-1)}+ \bar{\mu} h_{(\textcolor{black}{A_{1}^{\text{tr}}},B_m)},\label{For_A_3_11}\numberthis\\
\alpha_{2,1}=A_2+&w h_{(1,B_m)} + x h_{(1,B_m-1)}+ \bar{\mu} h_{(\textcolor{black}{A_{2}^{\text{tr}}},B_m)},\label{For_A_3_21}\numberthis
\end{align*}
\end{subequations}
\textcolor{black}{where $A_{i}^{\text{tr}}=\min\{A_{i}+1,A_{m}\}, \ i \in \{1,2\}$ due to the truncation of age.}
Assuming (\ref{equ:increasing_A}), one can see that, in all the above cases
\begin{equation} \label{alpha_10_leq_alpha_20___alpha_11_leq_alpha_21}
\alpha_{1,0} \leq \alpha_{2,0}, \ \alpha_{1,1} \leq \alpha_{2,1}.
\end{equation}
Then based on (\ref{alpha_10_leq_alpha_20___alpha_11_leq_alpha_21}), we can show
\begin{equation} \label{min_less_min_A}
\min\{\alpha_{1,0},\alpha_{1,1}\} \leq \min\{\alpha_{2,0},\alpha_{2,1}\}.
\end{equation}
Finally, based on (\ref{equ:min_hs_A}) and (\ref{min_less_min_A}) 
\begin{equation}
h^{k+1}_{(A_1,B)} \leq h^{k+1}_{(A_2,B)}, \ \forall B.
\end{equation}

Hence, the claim is proved, and the inequality (\ref{equ:increasing_A}) holds for all iterations, including the converged optimal cost function.
\hfill \qed

\subsection{Proof of Lemma \ref{lemma:increasing_B}}
Similar to the proof in Appendix \ref{App:lemma:increasing_A}, again, since the initialization $h(s), h_{\text{new}}(s) \leftarrow 0, \ \forall s \in S$ applies, (i) holds. For (ii), let $s_{1} \equiv s_{(A,B_1)}$ and $s_{2} \equiv s_{(A,B_2)}$ be two states such that $B_{1} \leq B_{2}$. Then, we define
\begin{equation} \label{beta}
\beta_{i,j} = A + \sum_{s^{'} \in S} p_{s_{i},d=j,s_{i}'} h^{k}(s_{i}'), \ i \in \{1,2\}, \ j \in \{0,1\}. 
\end{equation}
The updates for $s_1$ and $s_{2}$ are 
\begin{equation} \label{equ:min_hs_B}
h_{s_{1}}^{k+1}=\min\{\beta_{1,0},\beta_{1,1}\}, \ h_{s_{2}}^{k+1}=\min\{\beta_{2,0},\beta_{2,1}\}.
\end{equation} 
Now, we prove that under the assumption (\ref{equ:increasing_B}) we have $h^{k+1}_{(A,B_1)} \geq h^{k+1}_{(A,B_2)}$.

We compare the update terms in (\ref{beta}) corresponding to $j=0$ (no transmission) together and corresponding to $j=1$ (transmission) together, in four cases of the relative battery levels. For $\boldsymbol{B_1=0, \ 0<B_2<B_m}$
\begin{subequations} \label{For_B_1}
\begin{align*}
\beta_{1,0}=A+&\lambda h_{(\textcolor{black}{A^{\text{tr}}},1)} + \bar{\lambda} h_{(\textcolor{black}{A^{\text{tr}}},0)},\numberthis\label{For_B_1_10}\\
\beta_{2,0}=A+&\lambda h_{(\textcolor{black}{A^{\text{tr}}},B_2+1)} + \bar{\lambda} h_{(\textcolor{black}{A^{\text{tr}}},B_2)},\numberthis\label{For_B_1_20}\\
\beta_{1,1}=A+&w h_{(1,0)} + \bar{\lambda} h_{(\textcolor{black}{A^{\text{tr}}},0)} + y h_{(\textcolor{black}{A^{\text{tr}}},1)},\numberthis\label{For_B_1_11}\\
\beta_{2,1}=A+&w h_{(1,B_2)} + x h_{(1,B_2-1)} + z h_{(\textcolor{black}{A^{\text{tr}}},B_2)} \\
+&y h_{(\textcolor{black}{A^{\text{tr}}},B_2+1)},\label{For_B_1_21} \numberthis 
\end{align*}
\end{subequations}
for $\boldsymbol{B_1=0, \ B_2=B_m}$
\begin{subequations} \label{For_B_2}
\begin{align*}
\beta_{1,0}=A+&\lambda h_{(\textcolor{black}{A^{\text{tr}}},1)} + \bar{\lambda} h_{(\textcolor{black}{A^{\text{tr}}},0)},\label{For_B_2_10}\numberthis\\
\beta_{2,0}=A+& h_{(\textcolor{black}{A^{\text{tr}}},B_m)},\numberthis\label{For_B_2_20}\\
\beta_{1,1}=A+&w h_{(1,0)} + \bar{\lambda} h_{(\textcolor{black}{A^{\text{tr}}},0)}+ y h_{(\textcolor{black}{A^{\text{tr}}},1)},\numberthis\label{For_B_2_11}\\
\beta_{2,1}=A+&w h_{(1,B_m)} + x h_{(1,B_m-1)} + \bar{\mu} h_{(\textcolor{black}{A^{\text{tr}}},B_m)},\label{For_B_2_21}\numberthis
\end{align*}
\end{subequations}
for $\boldsymbol{0<B_1<B_m, \ 0<B_2<B_m}$
\begin{subequations} \label{For_B_3}
\begin{align*}
\beta_{1,0}=A+&\lambda h_{(\textcolor{black}{A^{\text{tr}}},B_1+1)} + \bar{\lambda} h_{(\textcolor{black}{A^{\text{tr}}},B_1)},\label{For_B_3_10}\numberthis\\
\beta_{2,0}=A+&\lambda h_{(\textcolor{black}{A^{\text{tr}}},B_2+1)} + \bar{\lambda} h_{(\textcolor{black}{A^{\text{tr}}},B_2)},\label{For_B_3_20}\numberthis\\
\beta_{1,1}=A+&w h_{(1,B_1)} + x h_{(1,B_1-1)} + z h_{(\textcolor{black}{A^{\text{tr}}},B_1)}\\
&+ y h_{(\textcolor{black}{A^{\text{tr}}},B_1+1)},\numberthis\label{For_B_3_11}\\ 
\beta_{2,1}=A+&w h_{(1,B_2)} + x h_{(1,B_2-1)} + z h_{(\textcolor{black}{A^{\text{tr}}},B_2)}\\
+& y h_{(\textcolor{black}{A^{\text{tr}}},B_2+1)},\label{For_B_3_21} \numberthis
\end{align*}
\end{subequations}
and for $\boldsymbol{0<B_1<B_m, \ B_2=B_m}$
\begin{subequations} \label{For_B_4}
\begin{align*}
\beta_{1,0}=A+&\lambda h_{(\textcolor{black}{A^{\text{tr}}},B_1+1)} + \bar{\lambda} h_{(\textcolor{black}{A^{\text{tr}}},B_1)},\label{For_B_4_10}\numberthis\\
\beta_{2,0}=A+& h_{(\textcolor{black}{A^{\text{tr}}},B_m)},\label{For_B_4_20}\numberthis\\
\beta_{1,1}=A+&w h_{(1,B_1)} + x h_{(1,B_1-1)} + z h_{(\textcolor{black}{A^{\text{tr}}},B_1)}\\
+ & y h_{(\textcolor{black}{A^{\text{tr}}},B_1+1)},\numberthis\label{For_B_4_11}\\
\beta_{2,1}=A+&w h_{(1,B_m)} + x h_{(1,B_m-1)} + \bar{\mu} h_{(\textcolor{black}{A^{\text{tr}}},B_m)},\label{For_B_4_21} \numberthis
\end{align*}
\end{subequations}
\textcolor{black}{where $A^{\text{tr}}=\min\{A+1,A_{m}\}$ due to the truncation of age.}
Given the assumption (\ref{equ:increasing_B}) and Lemma \ref{lemma:chain_inequality}, one can see that
\begin{equation} \label{beta_10_geq_beta_20___beta_11_geq_beta_21}
\beta_{1,0} \geq \beta_{2,0}, \ \beta_{1,1} \geq \beta_{2,1}
\end{equation}
holds for all the above four cases of the two battery levels.
Then based on (\ref{equ:min_hs_B}) and (\ref{beta_10_geq_beta_20___beta_11_geq_beta_21}), we can show
\begin{equation} \label{min_more_min_B}
\min\{\beta_{1,0},\beta_{1,1}\} \geq \min\{\beta_{2,0},\beta_{2,1}\}.
\end{equation}
Hence, according to (\ref{equ:min_hs_B}) and (\ref{min_more_min_B}), we can show 
\begin{equation}
h^{k+1}_{(A,B_{1})} \geq h^{k+1}_{(A,B_{2})}, \ \forall A.
\end{equation}
Therefore, the claim is established, and inequality (\ref{equ:increasing_B}) holds for all iterations, including the converged cost function.
\hfill \qed

\section{Proof of Theorem~\ref{theo:threshold_A} \textcolor{black}{and Discussion on Remark \ref{theo:threshold_B}}} \label{App:threshold}
\subsection{Proof of Theorem~\ref{theo:threshold_A}} \label{App:threshold_A}
Let us define
\begin{subequations} \label{theta}
\begin{align}  
\theta_{j}^* &= A + \sum_{s' \in S} p_{s,d=j,s'} h^*(s'), \ j \in \{0,1\},\label{theta_j}\\
\Delta \theta^*&=\theta_{1}^*-\theta_{0}^*. \label{exp:delta_theta}
\end{align}
\end{subequations}
From (\ref{exp:delta_theta}) and noting that the outcome of $\min\{\theta_{0}^{*},\theta_{1}^{*}\}$ determines the optimal action, we observe that
\begin{equation} \label{equ:d_s^*}
d_{s}^{*}=
\begin{cases}
0 & \Delta \theta^* \geq 0,\\
1 & \Delta \theta^* < 0.
\end{cases}
\end{equation}
To establish the threshold-based structural property of the optimal policy w.r.t. $A$, it suffices to show that $\Delta \theta$ is non-increasing in $A$. Then, (\ref{structure_A}) follows.
To show that $\Delta \theta$ is non-increasing in $A$, it suffices to show that $\Delta \alpha_{2}^{*} - \Delta \alpha_{1}^{*} \leq 0$, where
\begin{equation} \label{exp:delta_alpha_delta_alpha}
\Delta \alpha_{1}^{*}=\alpha_{1,1}^{*}-\alpha_{1,0}^{*}, \ \Delta \alpha_{2}^*=\alpha_{2,1}^*-\alpha_{2,0}^*,
\end{equation}
and where $\alpha_{i,j}^{*}$ is defined analogously to (\ref{alpha}) for the same states $s_{1}$ and $s_{2}$ defined there, and expressed across the three cases of the battery states similar to (\ref{For_A_1}), (\ref{For_A_2}), and (\ref{For_A_3}), all with $h$ replaced by $h^{*}$. Then $\Delta \alpha_{2}^{*} - \Delta \alpha_{1}^{*}$ is computed for these cases as follows.
For $\boldsymbol{B=0}$ 
\begin{equation} \label{exp:diff_Delta_For_A_B0}
\Delta \alpha_{2}^{*} - \Delta \alpha_{1}^{*} =\underbrace{w (h^{*}_{(\textcolor{black}{A_{1}^{\text{tr}}},1)}-h^{*}_{(\textcolor{black}{A_{2}^{\text{tr}}},1)})}_{-},
\end{equation}
for $\boldsymbol{0 < B < B_m}$
\begin{align*}
\Delta \alpha_{2}^{*} - \Delta \alpha_{1}^{*} =&\underbrace{w (h^{*}_{(\textcolor{black}{A_{1}^{\text{tr}}},B+1)}-h^{*}_{(\textcolor{black}{A_{2}^{\text{tr}}},B+1)})}_{-}\\
+&\underbrace{x (h^{*}_{(\textcolor{black}{A_{1}^{\text{tr}}},B)}-h^{*}_{(\textcolor{black}{A_{2}^{\text{tr}}},B)})}_{-}, \numberthis \label{exp:diff_Delta_For_A_0BBm}
\end{align*}
and for $\boldsymbol{B=B_m}$
\begin{equation} \label{exp:diff_Delta_For_A_Bm}
\Delta \alpha_{2}^{*} - \Delta \alpha_{1}^{*} =\underbrace{\mu (h^{*}_{(\textcolor{black}{A_{1}^{\text{tr}}},B_m)} -h^{*}_{(\textcolor{black}{A_{2}^{\text{tr}}},B_m)})}_{-}.
\end{equation}

Since $A_{1}\leq A_{2}$, by Lemma \ref{lemma:increasing}\ref{lemma:increasing_A}, we have $\Delta \alpha_{2}^{*} - \Delta \alpha_{1}^{*} \leq 0$ for all three expressions (\ref{exp:diff_Delta_For_A_B0}), (\ref{exp:diff_Delta_For_A_0BBm}), and (\ref{exp:diff_Delta_For_A_Bm}). Therefore, the optimal policy exhibits a threshold-based structure w.r.t. $A$, and (\ref{structure_A}) follows.
\hfill \qed

\subsection{\textcolor{black}{Discussion on Remark \ref{theo:threshold_B}}} \label{App:threshold_B}
Reconsidering (\ref{theta}) and (\ref{equ:d_s^*}), and following a similar approach to Appendix \ref{App:threshold_A}, to establish the threshold-based structural property of the optimal policy w.r.t. $B$, it must be shown that $\Delta \theta$ is non-increasing w.r.t. $B$. Then, (\ref{structure_B}) would follow directly.

To show that $\Delta \theta$ is non-increasing with $B$, we must show that $\Delta \beta_{2}^{*} - \Delta \beta_{1}^{*} \leq 0$, where
\begin{equation} \label{exp:delta_beta_delta_beta}
\Delta \beta_{1}^{*}=\beta_{1,1}^{*}-\beta_{1,0}^{*}, \ \Delta \beta_{2}^*=\beta_{2,1}^*-\beta_{2,0}^*,
\end{equation}
and where $\beta_{i,j}^{*}$ is defined similarly to (\ref{beta}) with the same $s_{1}$ and $s_{2}$ defined there, and expressed similarly as (\ref{For_B_1}), (\ref{For_B_2}), (\ref{For_B_3})\textcolor{black}{, and (\ref{For_B_4})} across four relative battery level cases, all with $h$ replaced by $h^{*}$.
For $\boldsymbol{B_1=0, \ 0<B_2<B_m}$
\begin{align*}
&\Delta \beta_{2}^{*} - \Delta \beta_{1}^{*}=\underbrace{w (h^{*}_{(\textcolor{black}{A^{\text{tr}}},1)}-h^{*}_{(\textcolor{black}{A^{\text{tr}}},B_2+1)})}_{+}\\
&+\underbrace{x (h^{*}_{(1,B_2-1)} -  h^{*}_{(\textcolor{black}{A^{\text{tr}}},B_2)})}_{?} +\underbrace{w (h^{*}_{(1,B_2)} -  h^{*}_{(1,0)})}_{-}, \numberthis  \label{exp:diff_Delta_For_B_1}
\end{align*}
for $\boldsymbol{B_1=0, \ B_2=B_m}$
\begin{align*}
&\Delta \beta_{2}^{*} - \Delta \beta_{1}^{*}= \underbrace{w (h^{*}_{(\textcolor{black}{A^{\text{tr}}},1)}-h^{*}_{(\textcolor{black}{A^{\text{tr}}},B_m)})}_{+}\\
&+\underbrace{x (h^{*}_{(1,B_m-1)} - h^{*}_{(\textcolor{black}{A^{\text{tr}}},B_m)})}_{?} + \underbrace{w (h^{*}_{(1,B_m)}- h^{*}_{(1,0)})}_{-}, \numberthis  \label{exp:diff_Delta_For_B_2}
\end{align*}
for $\boldsymbol{0<B_1<B_m, \ 0<B_2<B_m}$
\begin{align*}
&\Delta \beta_{2}^{*} - \Delta \beta_{1}^{*}=\\
&\underbrace{w (h^{*}_{(\textcolor{black}{A^{\text{tr}}},B_1+1)} - h^{*}_{(\textcolor{black}{A^{\text{tr}}},B_2+1)})}_{+} + \underbrace{x (h^{*}_{(\textcolor{black}{A^{\text{tr}}},B_1)}-  h^{*}_{(\textcolor{black}{A^{\text{tr}}},B_2)})}_{+}\\
&+ \underbrace{w (h^{*}_{(1,B_2)} - h^{*}_{(1,B_1)})}_{-} + \underbrace{x (h^{*}_{(1,B_2-1)} -  h^{*}_{(1,B_1-1)})}_{-}, \numberthis  \label{exp:diff_Delta_For_B_3} 
\end{align*}
and for $\boldsymbol{0<B_1<B_m, \ B_2=B_m}$
\begin{align*}
&\Delta \beta_{2}^{*} - \Delta \beta_{1}^{*}=\\
&\underbrace{w (h^{*}_{(\textcolor{black}{A^{\text{tr}}},B_1+1)} - h^{*}_{(\textcolor{black}{A^{\text{tr}}},B_m)})}_{+} + \underbrace{x (h^{*}_{(\textcolor{black}{A^{\text{tr}}},B_1)}-h^{*}_{(\textcolor{black}{A^{\text{tr}}},B_m)})}_{+}\\
&+ \underbrace{w (h^{*}_{(1,B_m)} - h^{*}_{(1,B_1)})}_{-} + \underbrace{x (h^{*}_{(1,B_m-1)} -  h^{*}_{(1,B_1-1)})}_{-}. \numberthis  \label{exp:diff_Delta_For_B_4}
\end{align*}
Given Lemmas \ref{lemma:increasing} and \ref{lemma:chain_inequality}, together with any other available mathematical relations and logical arguments, we are unable to establish that $\Delta \beta_{2}^{*} - \Delta \beta_{1}^{*} \leq 0$ holds for the expressions in (\ref{exp:diff_Delta_For_B_1}), (\ref{exp:diff_Delta_For_B_2}), (\ref{exp:diff_Delta_For_B_3}), \textcolor{black}{and (\ref{exp:diff_Delta_For_B_4})}, regardless of how the terms are arranged or grouped.\footnote{However, numerically we observe that $\Delta \theta$ is non-increasing in the second, third, and fourth cases. In the first case, it is occasionally non-monotonic; nevertheless, it does not change sign within that range, and the resulting policy still exhibits a threshold-based structure.}

\section{\textcolor{black}{Average AoA under AoI-Opt. (Greedy) Policy}} \label{App:Randomized_Policy_Proof}
Based on the expressions described in Remark \ref{remark:AoI_Opt_Dynamics}, a transition probability matrix, denoted by $\mathbf{T}$, can be obtained.
We denote the stationary distribution of the Markov chain by the row vector $\mathbf{V}$, where each element $V_{a,b}$ corresponds to the stationary probability of state $s_{(a=A,b=B)}$, with\footnote{\textcolor{black}{Here, since we work with Markov equations, we do not need to truncate the age, and we use the original range $A(k) \in \mathbb{N}$.}} $A(k) \in \mathbb{N}$ and $B(k) \in \{0,1,2,\dots,B_m\}$.
Also, $V_{A_{a}}=\sum_{b=0}^{B_{m}}V_{a,b}$ and $V_{B_{b}}=\sum_{a=1}^{\infty}V_{a,b}$ are the marginal probabilities for all the states with a certain AoA ($a$) and a certain battery level ($b$), respectively. 
Setting $\mathbf{V} \mathbf{T} = \mathbf{V}$, we can write the Markov equations
\begin{subequations} \label{RP_Markov_equs_4}
\begin{align}
V_{1,b}&=wV_{B_{b}}+xV_{B_{b+1}}, \ 0 \leq b \leq B_{m}, \label{RP_Markov_equs_4_a}  \\
V_{a,0}&=\bar{\lambda}V_{a-1,0}, \ a \geq 2, \label{RP_Markov_equs_4_b}\\
V_{a,b}&=yV_{a-1,b-1}+zV_{a-1,b}, \ a \geq 2,  1 \leq b \leq B_{m}-1, \label{RP_Markov_equs_4_c}\\
V_{a,\textcolor{black}{B_{m}}}&=\bar{\mu}V_{a-1,\textcolor{black}{B_{m}}}+yV_{a-1,\textcolor{black}{B_{m}}-1}, \ a \geq 2. \label{RP_Markov_equs_4_d}
\end{align}  
\end{subequations}

Summing $V_{1,b}$ in (\ref{RP_Markov_equs_4_a}) and $V_{a,b}$ in (\ref{RP_Markov_equs_4_b}), (\ref{RP_Markov_equs_4_c}), and (\ref{RP_Markov_equs_4_d}) over all $b$, we get
\begin{subequations} \label{equ:V_{A_{1}}_and_equ:V_{A_{2345...}}}
\begin{align}  
&V_{A_{1}}+xV_{B_{0}}=\mu, \label{equ:V_{A_{1}}}\\
&V_{A_{a}}=\bar{\mu} V_{A_{a-1}} +x V_{a-1,0}, \ a \geq 2. \label{equ:V_{A_{2345...}}}
\end{align}
\end{subequations}

Summing $V_{a,0}$, $V_{a,b}$ for $1 \leq b \leq B_{m}-1$, and $V_{a,\textcolor{black}{B_{m}}}$ in (\ref{RP_Markov_equs_4}) over all $a$, we get
\begin{subequations} \label{V_10_B_0_2_and_equ:V_{B_{0}}_2_2_and_equ:V_{B_{1}}_2_and_equ:V_{B_{m}}_2}
\begin{align}
V_{1,0}&=\lambda V_{B_{0}}, \label{V_10_B_0_2}\\
yV_{B_{0}}&=xV_{B_{1}}, \label{equ:V_{B_{0}_2}}\\
(x+y)V_{B_{b}}&=yV_{B_{b-1}}+xV_{B_{b+1}}, \  1 \leq  b \leq B_{m}-1, \label{equ:V_{B_{1}}_2}\\
xV_{B_{m}}&=y\textcolor{black}{V_{{B_{m}-1}}}. \label{equ:V_{B_{m}}_2}
\end{align}
\end{subequations}
By substituting $V_{B_{b}}$ terms into one another and using $\sum_{b=0}^{B_{m}}V_{B_{b}}=1$, we obtain
\begin{subequations} \label{equ:B_0_and_equ:B_b}
\begin{align}
V_{B_{0}}&=\textcolor{black}{\frac{1-r}{1-r^{B_{m}+1}}}, \label{equ:B_0} \\ 
V_{B_{b}}&= \frac{1-r}{1-r^{B_{m}+1}}r^{b}, \ b \in \{0,\cdots,\textcolor{black}{B_{m}}\}, \label{equ:B_b}
\end{align}
\end{subequations}
where $r=y/x$, \textcolor{black}{for $\mu \neq \lambda$. For $\mu=\lambda$, we have $r=1$, so all battery levels are equally likely and equal to
\begin{equation} \label{equ:B_0_singularity}
V_{B_{b}}=1/(B_{m}+1) , \ b \in \{0,\cdots,B_{m}\}.
\end{equation}}
\textcolor{black}{Note that $V_{B_{0}}=\bar{B}_{0}^{*,I}$.} From (\ref{equ:V_{A_{1}}_and_equ:V_{A_{2345...}}}), and (\ref{V_10_B_0_2}), we get 
\begin{equation} \label{equ:V_{A_{a}}_2}
V_{A_{a}}=(\mu-xV_{B_{0}})(\bar{\mu})^{a-1} +  \lambda x V_{B_{0}} \sum_{i=0}^{a-2} (\bar{\mu})^{a-2-i} \bar{\lambda}^{i}, \ a \geq 1.
\end{equation}
Then the average AoA is
\begin{subequations} \label{equ:Average_AoA_MC_App_General}
\begin{align}
\bar{A}^{*,I}=&\sum_{a=1}^{\infty}aV_{A_{a}}, \label{equ:Average_AoA_MC_App_a}\\
=&\frac{1}{\mu} - \frac{xV_{B_{0}}}{\mu^{2}} + x V_{1,0} \left( \frac{2}{\mu \lambda} + \frac{\bar{\lambda}}{\mu \lambda^{2} } +  \frac{\bar{\mu}}{\mu^{2} \lambda} \right), \label{equ:Average_AoA_MC_App_b}\\
=&\frac{1}{\mu} + \frac{\bar{\lambda}}{\lambda} V_{B_{0}}. \label{equ:Average_AoA_MC_App_c}
\end{align}
\end{subequations}
\textcolor{black}{Then for $\mu \neq \lambda$, from (\ref{equ:B_0}), (\ref{equ:Average_AoA_MC_App_c}) becomes
\begin{subequations} \label{equ:Average_AoA_MC_App_ForInequality}
\begin{align}
=&\frac{\mu x^{B_{m}+1}-\lambda y^{B_{m}+1}}{\mu\lambda\left(x^{B_{m}+1}-y^{B_{m}+1}\right)}, \label{equ:Average_AoA_MC_App_ForInequality_d}\\
=&\frac{\mu^{B_{m}+2}\left( 1- \lambda \right)^{B_{m}+1}-\left( 1- \mu \right)^{B_{m}+1} \lambda^{B_{m}+2}}{\mu^{B_{m}+2}\lambda\left( 1- \lambda \right)^{B_{m}+1}-\mu\left(1-\mu\right)^{B_{m}+1} \lambda^{\textcolor{black}{B_{m}}+2}}.  \label{equ:Average_AoA_MC_App_ForInequality_e}
\end{align}
\end{subequations}
And for $\mu=\lambda$, from (\ref{equ:B_0_singularity}), (\ref{equ:Average_AoA_MC_App_c}) becomes}
\begin{equation} \label{equ:Average_AoA_MC_App_ForEquality}
\bar{A}^{*,I}=\frac{2-\mu+B_{m}}{\mu+\mu B_{m}}.
\end{equation}
Thus, $\bar{A}^{*,I}$ equals (\ref{equ:Average_AoA_MC_App_ForInequality_e}) and (\ref{equ:Average_AoA_MC_App_ForEquality}), for $\mu \neq \lambda$ and $\mu=\lambda$, respectively.
\hfill \qed

\end{appendices}

\bibliographystyle{ieeetr}
\bibliography{bibliography}

\end{document}